\documentclass[acmsmall,nonacm]{acmart}
\usepackage{graphicx}
\usepackage[ruled,vlined,linesnumbered]{algorithm2e}
\usepackage{algpseudocode}
\usepackage{amsthm}
\usepackage{xspace}
\usepackage{wrapfig}

\usepackage{amsthm, amsmath, mathtools}
\usepackage{empheq}

\usepackage{upgreek}
\usepackage{amsfonts}
\usepackage{enumitem}
\usepackage[normalem]{ulem}
\usepackage[para,online,flushleft]{threeparttable}
\usepackage{multirow}
\usepackage{multicol}
\usepackage{hyperref}
\usepackage{url}
\usepackage{bm}
\usepackage[labelfont=bf,skip=0pt]{caption,subfig}
\usepackage{flushend}

\usepackage{tabularx}

\usepackage[english]{babel}
\usepackage{graphicx}
\usepackage{thmtools}

\usepackage{listings}
\usepackage{tablefootnote}
\usepackage{enumitem}
\usepackage{bm}
\usepackage{comment}
\usepackage{tikz}

\DeclareMathOperator*{\argmax}{arg\,max}
\DeclareMathOperator*{\argmin}{arg\,min}

\newcommand{\Q}{\boldsymbol{q}}

\newcommand{\A}{\mathcal{A}}

\newcommand{\U}{\mathcal{U}}

\newcommand{\eps}{\varepsilon}
\renewcommand{\epsilon}{\varepsilon}
\renewcommand{\Re}{\mathbb{R}}
\newcommand{\sphere}{\mathbb{S}}
\newcommand{\qtime}{\mathbb{T}}
\renewcommand{\paragraph}[1]{\smallskip\noindent{\bf {#1. }}}

\newcommand{\dom}{{\texttt{dom}}}

\newcommand{\query}{\Q}
\newcommand{\D}{\mathbf{D}}
\newcommand{\I}{\mathcal{I}}
\newcommand{\net}{\mathcal{C}}

\newcommand{\tree}{\mathcal{T}}
\newcommand{\node}{u}

\newcommand{\prob}{\textsf{OracleCluster}{}}

\allowdisplaybreaks
\usepackage{xcolor}
\usepackage{amsmath}

 \usepackage{tcolorbox}

\usepackage{mdframed}

\newcommand{\DS}{\points}
\newcommand{\Rep}{\mathcal{P}}
\newcommand{\Synop}{\mathcal{S}}

\newcommand{\polylog}{\text{\,}\mathrm{polylog}\,}

\usepackage{bbm}

\usepackage[mathscr]{euscript}

\newcommand{\F}{\mathcal{F}}

\newcommand{\M}{\mathcal{M}}

\newcommand{\err}{\textsc{Err}}

\newcommand{\pred}{\textsc{Pred}}

\newcommand{\new}[1]{#1}

\newtheorem{problem}{Problem}

\newcommand{\sets}{\mathcal{P}}
\newcommand{\setsize}{N}
\newcommand{\points}{P}
\newcommand{\pointsize}{n}
\newcommand{\totalsize}{\mathcal{N}}
\newcommand{\rect}{R}
\newcommand{\rec}{\rho}
\newcommand{\interval}{\theta}
\newcommand{\intervalL}{a_{\theta}}
\newcommand{\intervalU}{b_{\theta}}
\renewcommand{\prob}{\varphi}
\newcommand{\epsample}{\mathsf{S}}
\newcommand{\rangetree}{\mathcal{T}}
\newcommand{\out}{|J|}
\newcommand{\outt}{\mathsf{OUT}}
\renewcommand{\vector}{\vec{v}}
\renewcommand{\O}{\tilde{O}}
\newcommand{\threshold}{\tau}
\newcommand{\score}{\omega}
\newcommand{\kernel}{\mathsf{S}}
\newcommand{\bbdtree}{\mathcal{T}}

\newcommand{\conj}{m}
\newcommand{\probTopk}{preference}
\newcommand{\timeSynop}{\Lambda}

\newcommand{\problemDI}{\textsf{Ptile}\xspace}
\newcommand{\problemCDI}{\textsf{CPtile}\xspace}
\newcommand{\problemFDI}{\textsf{FPtile}\xspace}

\newcommand{\problemDIk}{\textsf{Pref}\xspace}
\newcommand{\problemCDIk}{\textsf{CPref}\xspace}
\newcommand{\problemFDIk}{\textsf{FPref}\xspace}

\begin{document}
\title{A Theoretical Framework for Distribution-Aware Dataset Search}

\author{Aryan Esmailpour}
\affiliation{%
  \institution{Department of Computer Science, University of Illinois Chicago}
  \city{Chicago}
  \country{USA}}
\email{aesmai2@uic.edu}

\author{Sainyam Galhotra}
\affiliation{%
  \institution{Department of Computer Science, Cornell University}
  \city{Ithaca}
  \country{USA}}
\email{sg@cs.cornell.edu}

\author{Rahul Raychaudhury}
\affiliation{%
  \institution{Department of Computer Science, Duke University}
  \city{Durham}
  \country{USA}}
\email{rahul.raychaudhury@duke.edu}

\author{Stavros Sintos}
\affiliation{%
  \institution{Department of Computer Science, University of Illinois Chicago}
  \city{Chicago}
  \country{USA}}
\email{stavros@uic.edu}

\begin{abstract}
Effective data discovery is a cornerstone of modern data-driven decision-making. Yet, identifying datasets with specific distributional characteristics, such as percentiles or preferences, remains challenging. While recent proposals have enabled users to search based on percentile predicates, much of the research in data discovery relies on heuristic methods, which often result in biased outcomes. This paper presents the first theoretically backed framework that unifies data discovery under centralized and decentralized settings. 

More specifically, let $\sets=\{\points_1,\ldots, \points_\setsize\}$ be a 
repository of $\setsize$ datasets, such that each $\points_i\subset \Re^d$, where $d$ is a constant. 
We study the percentile-aware indexing (\problemDI) problem and the \probTopk-aware indexing  (\problemDIk) problem under the centralized and the federated setting.
In the centralized setting, we assume direct access to the datasets in $\sets$. In the federated setting we are given a synopsis $\Synop_{\points_i}$ which is a compressed representation of $\points_i$ that captures the structure of $\points_i$, for every $i\in[\setsize]$.
For the \problemDI problem, the goal is to construct a data structure such that given a predicate (query rectangle $\rect$ and an interval $\interval$) report all indexes $J$ such that $j\in J$ if and only if $\frac{|\points_j\cap \rect|}{|\points_j|}\in\interval$.
For the \problemDIk problem, the goal is to construct a data structure such that given a predicate (query vector $\vector$ and an interval $\interval$) report all indexes $J$ such that $j\in J$ if and only if $\score_k(\points_j,\vector)\in \interval$, where $\score_k(\points_j,\vector)$ is the score (inner-product) of the $k$-th largest projection of $\points_j$ on $\vector$.
We first show lower bounds for the \problemDI and \problemDIk problems in the centralized setting, showing that we cannot hope for near-linear data structures with polylogarithmic query time. Then we focus on approximate data structures for both problems in both settings. We show $\O(\setsize)$ space data structures with $\O(\setsize)$ preprocessing time, that can answer \problemDI and \problemDIk queries in $\O(1+\outt)$ time, where $\outt$ is the output size. The data structures return a set of indexes $J$ such that: i) for every $\points_i$ that satisfies the predicate, $i\in J$ and ii) if $j\in J$ then $\points_j$ satisfies the predicate up to an additive error of $\eps+2\delta$, where $\eps$ is an arbitrarily small constant and $\delta$ is the error of the synopses.
\end{abstract}

\maketitle



\section{Introduction}
\label{sec:intro}

In today's data-driven world, efficient and effective data discovery~\cite{castelo2021auctus, castro_fernandez_aurum_2018, galhotra_metam_2023,bharadwaj2021discovering, bogatu2020dataset, esmailoghli2022mate, gong2023ver, koutras_valentine_2021, miller2018making, zhang_dsdd_2021, zhang_finding_2020, gong2024nexus} plays a pivotal role in unlocking the true potential of information collected from diverse sources. With the rise of deep learning, organizations increasingly rely on vast and varied datasets to drive innovation and gain competitive advantage.

However, existing data discovery systems 
primarily rely on keyword-based searches or examples to identify relevant datasets~\cite{castelo2021auctus, noy_google_2019}. 
While these methods are ubiquitous, they fall short in scenarios where users have a specific distributional requirement. For instance, users developing machine learning models often require datasets that ensure balanced representation across various groups to prevent overfitting or biases in their models.  
Furthermore, the problem of selection bias (lack of representative samples) often leads to flawed analyses and unreliable outcomes. These shortcomings are particularly pronounced in industrial settings~\cite{dastin_amazon_2022, mulshine_major_2015, rose_are_2010, townsend_most_2017}, where datasets are frequently repurposed or reused in contexts different from their original design~\cite{bethlehem_selection_2010, culotta_reducing_2014, greenacre_importance_2016, zhu_consistent_2023}.  Therefore, it is important to develop advanced data discovery techniques that allow users to search for datasets that satisfy their distributional needs. We demonstrate this with the following example.

\begin{example} \label{example1}
\new{Consider an economist who wants to understand the various reasons behind high crime in Brooklyn, NY (more generally, in any specific neighborhood in the US). To ensure statistical significance of the analysis, the economist searches for data that contain a substantial fraction of the data points (say at least 10\%) from Brooklyn.
More generally, such support requirement would hold even if they want to analyze crime (or other factors affecting development) across different cities across the United States. Several prior works~\cite{behme2024fainder,gong2024nexus,galhotra_metam_2023,gong2023ver} on data discovery have considered such queries for data analysis on open data repositories which contain roughly 100K datasets. }


\new{Additionally, the economist may want to identify cities which have at least k neighborhoods with high or low quality of life, modeled as a linear function of crime, pollution, healthcare and other factors that they are interested to consider.}


    \new{In summary, the economist often faces two types of distributional requirements for  data search: (i) The fraction of data points within a specific range (geographical region)  must meet or exceed a specified threshold, such as x\%. We refer to this query as a percentile-query.
    (ii) The value (or score) of a specific combination of attributes from at least $k$ data points must exceed a threshold, such as $\threshold$. We formalize this notion as a preference query.}

\end{example}

The prior work on distribution aware search suffer from two main challenges:
First, these techniques cannot handle preference requirements and only support one-sided percentile queries. For example, a query to identify datasets that contain at least $10\%$ points in a rectangular area (that includes Brooklyn) can not be answered with these tools.
Second, all these techniques employ heuristics to identify the best datasets. More generally, prior work on  data discovery employs heuristics (for example in keyword based data discovery techniques~\cite{castelo2021auctus, noy_google_2019}), which do not guarantee the effectiveness of the search algorithm. 
This means that the search result may miss certain datasets, even though it satisfies user's query. Missing any dataset that satisfies user query can percolate biases into the system. 
Missing datasets is generally unacceptable in data marketplaces, where it is extremely crucial to ensure that all datasets that satisfy the query are available to the user. Therefore, it is important to develop a theoretically rigorous framework for data discovery, which guarantees to return all relevant datasets for the user's query.


In this work, we address these limitations and present the first theoretical framework for data discovery which helps users to identify datasets that satisfy distributional requirements.
\new{Next, we describe a new general theoretical framework that is used to formally define various practical problems from distribution-aware dataset search. Later, we show two of the most common subclasses of problems that we focus in this paper, derived from our new framework.}


\subsection{\new{Theoretical Framework}}
\label{subsec:framework}
\new{In this section, we present a new theoretical framework for data discovery.}

\paragraph{Dataset}
A \emph{schema} $\A = (A_1, A_2, \ldots, A_d)$ is an ordered tuple of attributes. A \emph{dataset} $\DS$ with schema $\A$ is a finite set of data entries, where each entry $p\in \DS$ is a $d$-tuple $p = (p_1, p_2, \ldots, p_d)$ such that $p_i \in \dom(A_i)$ for $i = 1, 2, \ldots, d$. Thus, $\DS \subseteq \dom(A_1) \times \dom(A_2) \times \cdots \times \dom(A_d)$. In this paper, we focus on datasets where all attributes are numerical, i.e., $\dom(A_i) = \mathbb{R}$ for all $i$. Hence, $\DS \subset \mathbb{R}^d$. We denote the schema of a dataset $\DS$ as $\textsc{Schema}(\DS)$. In all cases, we assume that $d=O(1)$.

\paragraph{Repository} A \emph{repository} \(\Rep\) is a collection of datasets \(\Rep = \{\DS_1, \DS_2, \ldots, \DS_{\setsize}\}\), all of which share the same schema, i.e., \(\textsc{Schema}(\DS_i) = \textsc{Schema}(\DS_j)\) for all \(i, j\). \footnote{The assumption that all datasets contain only numerical attributes and share the same schema is not strictly required. These assumptions are adopted to simplify the presentation. For our framework to be applicable, it suffices for datasets in $\Rep$ to share at least some numerical attributes. }
For each $i\in[\setsize]$, let $\pointsize_i=|\DS_i|$ be the size of dataset $\DS_i$, and $\totalsize=\sum_{i\in[\setsize]}\pointsize_i$.

\paragraph{Measure Function}  A \emph{measure function} \(\M(\cdot)\), is a function that, when applied to a dataset \(\DS\), assigns a numerical value \(\M(\DS) \in \mathbb{R}\) to quantify a specific property or characteristic of the dataset. We assume that \(\M(\cdot)\) is applied to a dataset \(\DS\) only if \(\M(\DS)\) is well-defined. For example, a commonly used measure function is the percentile measure function $\M_{\rect}$ over a rectangle $\rect$, which is defined as $\M_{\rect}(\DS)=\frac{|\rect\cap \DS|}{|\DS|}$.


A set of measure functions that quantify similar properties is called a \emph{class}, which we typically denote by \(\F\). For example, a the class of percentile measure functions consists of all measure functions $\M_{\rect}$ such as $\rect$ is any axis-parallel rectangle in $\Re^d$.





\paragraph{Predicates} A \emph{predicate} \(\pred(\cdot)\) is a function that, given a dataset \(\DS\), evaluates to \(\pred(\DS) \in \{\textsf{True}, \textsf{False}\}\). For a measure function \(\M\) and an interval \(\theta = [\intervalL, \intervalU]\), where \(\intervalL, \intervalU \in \mathbb{R}\) and \(\intervalL \leq \intervalU\), a \emph{range-predicate} \(\pred_{\M, \theta}(\DS)\) returns \textsf{True} if \(\M(\DS) \in \theta\). If $\theta$ is a one-sided interval (for instance, $\interval=[\intervalL,\infty)$), we call \(\pred_{\M, \theta}(\cdot)\) a \emph{threshold-predicate}. 
Complex predicates can be formed by combining individual predicates using conjunction and, disjunctions.
For example, consider the complex predicate \(\Pi = \pred_{\M,\interval} \land \pred_{\M',\interval'}\), over two measure functions $\M$ and $\M'$ of the same class. The expression \(\Pi(\DS)\) evaluates to \textsf{True} if both \(\pred_{\M,\interval}(\DS)\) and \(\pred_{\M',\interval'}(\DS)\) are \textsf{True}.




\paragraph{Synopsis}
Let $\F$ be a class of measure functions. For a dataset $\DS$, a \emph{synopsis} (or \emph{sketch}) with respect to the class $\F$, denoted by $ \Synop_{\DS}^{\F}$, is a compressed representation of $\DS$ that captures the structure of $\DS$ with regards to the class $\F$. 
For example, the class of percentile measure functions is defined as the collection of functions $\M_\rect$ over every rectangle $\rect$.
Ideally, for any function $\M \in \F$, \(
\M(\DS) \approx \M(\Synop_{\DS}^{\F}),
\)  where $\M(\Synop_{\DS}^{\F})$ is the result of evaluating the measure function $\M$ on $\Synop_{\DS}^{\F}$. \footnote{We assume $\M(\Synop_{\DS}^{\F})$ is well-defined.} 
We ignore the superscript and only write $\Synop_{\DS}$, when the class $\F$ is clear from context.
For a function $\M\in\F$, we denote the error of $\Synop_{\DS}$ on $\M$ as $\err_{\Synop_{\DS}}(\M)$, where the specific definition of  $\err_{\Synop_{\DS}}(\cdot)$ depends on the class $\F$ and the dataset $\DS$. We define $\err_{\Synop_{\DS}}(\F):=\max_{\M\in\F}\err_{\Synop_{\DS}}(\M)$.

\paragraph{Distribution-Aware Indexing} Let \(\Rep = \{\DS_1, \DS_2, \ldots, \DS_{\setsize}\}\) be a repository, where each dataset \(\DS_i \subset \mathbb{R}^d\), and let \(\F\) be a class of measure functions. The general \emph{distribution-aware indexing problem} for $(\Rep,\F)$ aims to design an index \(\I\) that, given a logical expression \(\Pi\) (of constant size) involving predicates of the form \(\pred_{\M, \theta}\), where \(\M \in \F\) and \(\theta\) is an interval in \(\mathbb{R}\), efficiently computes the subset of datasets in \(\Rep\) that satisfy \(\Pi\). Formally, the index \(\I\) should answer a \emph{distribution-aware dataset query} \(\query_{\Pi}\):
\vspace{-1em}
\[
\query_\Pi(\Rep):= \{i \in [\setsize]\mid \Pi(\DS_i) = \textsc{True}\}.
\]

\par There are two natural variants of the problem.
In the \emph{centralized} setting, 
we have direct and full access to all the datasets in \(\Rep\) for building \(\I\).
In the \emph{federated} setting,
we do not have direct access to \(\Rep\). Instead, we have access to a collection of synopses \(\Rep_{\Synop} = \{\Synop_{\DS_1}, \Synop_{\DS_2}, \ldots, \Synop_{\DS_{\setsize}}\}\), where \(\Synop_{\DS_i}\) is a synopsis of \(\DS_i\) with respect to \(\F\). We note that if $\Synop_{\points_i}=\points_i$, for every $i\in[\setsize]$, then this is equivalent to the centralized setting.

\par It is evident that, \(\query_{\Pi}(\Rep)\) can be answered exactly in the centralized setting, where direct access to all datasets in \(\Rep\) is available. In the federated setting, however, the ability to evaluate  \(\Pi(\DS_i)\), for any $i$, critically depends on the quality of the synopsis \(\Synop_{\DS_i}\) provided for each dataset \(\DS_i\). Specifically, the ability to correctly decide any predicate \(\pred_{\M_j, \theta_j}(\DS_i)\) is influenced by the error \(\err_{\Synop_{\DS_i}}(\M)\), which is a measure of deviation between \(\M(\Synop_{\DS_i})\) and  \(\M(\DS_i)\).


\subsection{\new{Problem definition}}
\label{subsec:probldef}
In order to define a problem based on the proposed framework in Section~\ref{subsec:framework}, someone needs to specify the class of measure functions.
In this paper, we propose efficient data structures for two of the most common classes of measure functions in the distribution-aware indexing problem.

\paragraph{Percentile-aware indexing} Let \(\DS \subset \mathbb{R}^d\) be a dataset and let \(\rect\) be an axis-parallel rectangle in \(\mathbb{R}^d\). A \emph{percentile measure function} \(\M_{\rect}\) is defined as:
$\M_\rect(\DS) := \frac{| \rect \cap \DS |}{| \DS |}$.
Interpreting \(\DS\) as a uniform discrete probability distribution with support defined by the points in $\DS$, \(\M_\rect(\DS)\) measures the mass of the distribution in the region \(\rect\). A percentile measure function is relevant when the user is interested in the proportion of tuples within a specific region of a dataset, e.g., the proportion of tuples from Brooklyn.
We denote the class of axis-parallel percentile measure functions in $\Re^d$ by:
\[
\F^d_{\Box} = \{\M_{\rect}(\cdot) \mid \rect \text{ is an axis-aligned rectangle in } \mathbb{R}^d\}.
\]
For the class $\F^d_{\Box}$, the commonly used synopses in practice include probability distributions over \(\mathbb{R}^d\), such as histograms, mixture models, and \(\varepsilon\)-samples
 (a formal definition of an $\eps$-sample is given in the next section).
 Such synopses support random sampling over the compressed representation of their datasets. 
The error in this context is defined as
$\err_{\Synop_{\points_i}}(\M_{\rect}) := |\M_{\rect}(\points_i) - \M_{\rect}(\Synop_{\points_i})|,$
where \(\M_{\rect}(\Synop_{\points_i}) := \Pr_{p \sim \Synop_{\points_i}}[p \in \rect]\).
For simplicity, we assume that for every $i\in[\setsize]$,
$\err_{\Synop_{\points_i}}(\F^d_{\Box})\leq\delta_i\leq \delta$, where $\delta\in[0,1)$ is a known arbitrarily small constant.
For the class of measure functions $\F^d_{\Box}$, the problem derived by our new framework is better known in the literature as \emph{percentile-aware indexing} \textbf{(\problemDI)} \cite{behme2024fainder,asudeh2022towards}.

\begin{problem}[\problemDI]
\label{prob1}
\new{    Let $\Rep$ be a repository over a set of (possibly) unknown datasets $\{\DS_1,\ldots, \DS_N\}$, where $\DS_i\subset \Re^d$ for every $i\in[N]$ and $d=O(1)$.
We are given a family of $N$ synopses $\{\Synop_{\DS_1},\ldots, \Synop_{\DS_N}\}$, with regards to the class of measure functions $\F^d_{\Box}$, such that $\err_{\Synop_{\points_i}}(\F^d_{\Box})\leq \delta$ for every $i\in[\setsize]$, for $\delta\in [0,1)$. The goal is to construct a data structure such that, given a logical expression $\Pi$ over a constant number of range-predicates of the form $\pred_{\M_\rect,\interval}$, where $\M_\rect\in \F^d_{\Box}$, it returns $\query_\Pi(\Rep)$.}
\end{problem}
In the centralized setting ($\delta=0$), we call the problem the \emph{centralized percentile-aware indexing} \textbf{(\problemCDI)}, and in the federated setting ($\delta>0$), \emph{federated percentile-aware indexing} \textbf{(\problemFDI)}.

 \medskip
\paragraph{Preference-aware Indexing} Let \(\DS \subset \mathbb{R}^d\) be a dataset. For a unit vector \(\vector \in \mathbb{R}^d\) and a point \(p \in \DS\), define the score of point $p$ with respect to vector $\vector$ as \(\score(p, \vector) = \langle p, \vector \rangle\), i.e., the inner product of \(p\) with \(\vector\). For an integer \(k > 0\), let \(\score_k(\DS, \vector)\) denote the score of the point in \(\DS\) with the \(k\)-th largest inner-product with respect to \(\vector\). A \emph{top-k preference measure function} is defined as:
$\M_{\vector,k}(\DS) := \score_k(\DS, \vector)$.
A top-\(k\) preference measure function is relevant when the user is interested in the top-\(k\) scores of a dataset with respect to some linear function over the attributes, e.g., cities with $k$ neighborhoods with high/low quality of life, modeled as a linear function.
We denote the class of top-k preference measure functions in $\Re^d$ by 
\vspace{-0.5em}
\[
\vec{\F}^d_k = \{\M_{\vector,k}(\cdot) \mid \vector \text{ is a unit-vector in } \mathbb{R}^d\}.
\vspace{-0.4em}
\]
For the class $\vec{\F}^d_{k}$, then a commonly used synopsis in practice include a kernel~\cite{agarwal2017efficient, kumar2018faster, yu2012processing} or a histogram.
Such synopses have a mechanism to return
$\M_{\vector,k}(\Synop_\points)$, i.e., an estimation of the $k$-th largest score in dataset $\points$ with respect to unit vector $\vector$.
The error in this context is also defined as
$\err_{\Synop_{\DS}}(\M_{\vector,k}) := |\M_{\vector,k}(\DS)-\M_{\vector,k}(\Synop_{\DS})|$.
For simplicity, we assume that for every $i\in[\setsize]$,
$\err_{\Synop_{\points_i}}(\vec{\F}^d_{k})\leq\delta_i\leq \delta$, where $\delta\in[0,1)$ is a known arbitrarily small constant.
For the class of measure functions $\vec{\F}^d_{k}$, the problem derived by our new framework is called \emph{\probTopk-aware indexing} \textbf{(\problemDIk)}.

\begin{problem}[\problemDIk]
\label{prob2}
\new{    Let $\Rep$ be a repository over a set of (possibly) unknown datasets $\{\DS_1,\ldots, \DS_N\}$, where $\DS_i\subset \Re^d$ for every $i\in[N]$ and $d=O(1)$. Let $k$ be a positive integer number.
We are given a family of $N$ synopses $\{\Synop_{\DS_1},\ldots, \Synop_{\DS_N}\}$, with regards to the class of measure functions $\vec{\F}^d_{k}$, such that $\err_{\Synop_{\points_i}}(\vec{\F}^d_{k})\leq \delta$ for every $i\in[\setsize]$, for $\delta\in [0,1)$. The goal is to construct a data structure such that, given a logical expression $\Pi$ over a constant number of threshold-predicates of the form $\pred_{\M_{\vector,k},\interval}$, where $\M_{\vector,k}\in \vec{\F}^d_{k}$, it returns $\query_\Pi(\Rep)$.}
\end{problem}
In the centralized setting ($\delta=0$), we call the problem \emph{centralized \probTopk-aware indexing} \textbf{(\problemCDIk)}, and in the federated setting ($\delta>0$) \emph{federated \probTopk-aware indexing} \textbf{(\problemFDIk)}.

\subsection{Our Results}\label{sec:results}

We use $\O(\cdot)$ notation to suppress exact dependencies on $\log \setsize$ and $\log \totalsize$. Detailed dependencies are fully specified in subsequent sections.
In all cases, we denote the set of reported indexes (corresponding to datasets in \(\sets\)) as \(J\), and sometimes refer to \(\outt = |J|\) as the output size.

\medskip
$\bullet$ (\textbf{Section~\ref{sec:lb}: Lower Bounds})
In Subsection~\ref{sec:lower-CDAI}, we focus on \problemCDI and rely on the conjectured hardness of the \emph{strong set intersection problem}  to show that no data structure with \(\O(\totalsize)\) size and \(\O(1+\outt)\) query time exists in \(\Re^2\), even with a single threshold-predicate (conditional lower bound). 
In Subsection~\ref{sec:lower:CPAI}, using the hardness of the \emph{halfspace reporting problem}, we show that no data structure of \(\O(\totalsize)\) size and \(\O(1+\outt)\) query time exists for the \problemCDIk problem in \(\Re^5\), even with one predicate (unconditional lower bound). 

\medskip

$\bullet$ (\textbf{Section~\ref{sec:Perc}: Approximate Data structures for \problemDI}) 
\new{
We first, design an approximate data structure for the \problemDI problem in \(\Re^d\), supporting a single range-predicate.
It has \(\O(\setsize)\) space, \(\O(\setsize)\) preprocessing time, and \(\O(1+\outt)\) query time.
For a query \(\Pi = \pred_{\M_{\rect}, \interval}\), with \(\rect\) as a query rectangle and \(\interval = [\intervalL, 1]\), the data structure returns \(J\) such that \(\query_\Pi(\sets) \subseteq J\), and if \(j \in J\), then  \(\intervalL - \eps - 2\delta \leq \M_{\rect}(\points_j) \leq \intervalU + \eps + 2\delta\), with high probability, where \(\eps \in (0,1)\) is a small constant.
Next, we extend the data structure to work for general intervals $\interval=[\intervalL, \intervalU]$ (range-predicate), having the same error and the same theoretical guarantees.}
 In Appendix~\ref{subsec:percConj}, we extend the structure further to handle logical expressions (supporting disjunctions and conjunctions) of $\conj$ range-predicates, where \(\conj = O(1)\), while retaining the same error bounds and same asymptotic complexities.
\medskip
$\bullet$ (\textbf{Section~\ref{sec:topK}: Approximate Data structures for \problemDIk})
We design an approximate data structure that works for the \problemDIk problem assuming one threshold-predicate. It has \(\O(\setsize)\) space, \(\O(\setsize)\) preprocessing time, and \(\O(1+\outt)\) query time. For a query \(\Pi = \pred_{\M_{\vector,k}, \interval}\), where \(\vector\) is a query unit vector and $\interval$ is a one-sided query interval, it returns a set \(J\) such that \(\query_\Pi(\sets) \subseteq J\), and if \(j \in J\), then \(\M_{\vector, k}(\points_j) \geq \intervalL - \eps - 2\delta\). In Appendix~\ref{subsec:topkpredm}, we extend this data structure to handle logical expressions of $\conj$ threshold-predicates, where \(\conj = O(1)\), while retaining the same error bounds and same asymptotic complexities.



\medskip 
$\bullet$ (\textbf{Additional Results})
All of our data structures satisfy the following properties: i) They extend to settings where \(\delta_i\)'s are unknown, with no global upper bound \(\delta\), while maintaining the same guarantees.
ii) They are dynamic, supporting \(\O(1)\) update time for the insertion or deletion of a synopsis. iii) They provide \(\O(1)\) delay guarantees, ensuring the time between reporting two consecutive indexes is \(\O(1)\).

\subsection{Related work}

In traditional database systems, there are two settings for query-driven dataset search: data lakes and metadata-based search.
Data lakes usually assume full data access and their proposed solutions do not have theoretical guarantees~\cite{bharadwaj2021discovering, bogatu2020dataset, esmailoghli2022mate, fernandez2018aurum, gong2024nexus, gong2023ver, miller2018making,nargesian_data_2023, nargesian_organizing_2020, ouellette_ronin_2021}.
Existing metadata-based dataset search engines, such as~\cite{brickley2019google, castelo2021auctus}, are great engineering solutions but they do not provide theoretical guarantees or/and are limited to keyword search.


Asudeh and Nargesian~\cite{asudeh2022towards} developed a system vision for distribution-aware data discovery. While their vision targets the data lake setting, it highlights the importance of distributional queries over dataset collections. Similarly, Nargesian et al.~\cite{nargesian2021tailoring} proposed distribution tailoring to meet fairness requirements in the same setting.
These methods require raw data access and they do not provide theoretical error guarantees for our studied problems.

The definition of \problemDI problem was first given in~\cite{behme2024fainder}, which built a federated search technique assuming that histograms represent the synopsis of each dataset $\points_i$.
While their method detects all datasets whose percentile measure function is in $\interval$, its query time is super linear with respect to the number of datasets.
Furthermore, they only studied one-sided query predicates $\interval$ and their method does not extend to range predicates.

Our problems for the centralized setting can be modeled as range color reporting queries. Indeed, if every dataset has a different color, then the \problemCDI problem asks to report all colors in a query rectangle $\rect$ that contain a fraction (at least $\intervalL$ and at most $\intervalU$) of their points in $\rect$. Similarly, the \problemCDIk problem asks to report all colors such that the score of their $k$-th larger point on a vector $\vector$ is at least $\intervalL$. The paper closest to our work is~\cite{afshani2023range}, where the authors study range summary queries on colored points. Given a set of $n$ colored points $P\subset\Re^d$ the goal is to construct a data structure such that given a query rectangle $\rect$, it reports all colors that appear at least $\eps|P\cap \rect|$ times in $\rect$. There are fundamental differences with our problem definition and the proposed data structures. i) Our \problemCDI problem aims to report an index (color) $j$ if $\frac{|\points_j\cap \rect|}{|\points_j|}\in \interval$. In their setting they only handle one sided intervals, and $j$ is reported if $|\points_j\cap \rect|\geq \eps\cdot \sum_{i\in[\setsize]}|\points_i\cap \rect|$. 
It is not clear how the data structures in~\cite{afshani2023range} can be used to solve our problems efficiently.
ii) They design data structures only for $d\leq 3$.
iii) They do not consider the \problemFDI or the \problemFDIk problems.

Furthermore, there are more works on range queries in colored points such as, color reporting~\cite{gupta1995further,chan2020further}, color counting~\cite{kaplan2007counting}, approximate color counting~\cite{rahul2017approximate, nekrich2014efficient}, or more complex queries on colors such as entropy queries~\cite{krishnan2024range, esmailpour2024range}.
However, the techniques of these papers do not apply to our problem.
Lu and Tao~\cite{lu2023indexing} studied the orthogonal range reporting problem with keywords. 
Our studied problems are fundamentally different than the problem in~\cite{lu2023indexing} and it is not clear how their methods can be used for \problemDI and \problemDIk problems under the centralized and federated setting.

\section{Preliminaries}
\label{sec:prelim}
\paragraph{Hyper-rectangles} A hyper-rectangle $\rect$ in $\Re^d$ is usually defined in two ways. Either $\rect$ is defined by its two opposite corners $\rect^-, \rect^+\in \Re^d$, or $\rect$ is defined as the product of $d$ intervals $\rect=[\rect^-_1,\rect^+_1]\times\ldots\times [\rect^-_d,\rect^+_d]$, where $\rect^-_j, \rect^+_j$ are the $j$-th coordinates of the points $\rect^-, \rect^+$, respectively. For $d=2$, $\rect$ is a rectangle, and for $d=1$, $\rect$ is an interval. For simplicity, for every $d\geq 1$ we call $\rect$ a rectangle.
An open rectangle $\rect$ that is only defined by one of its corners is called an \emph{orthant}.

\paragraph{$\eps$-sample}
Let $\mathcal{S}=(\mathbf{X}, \mathcal{R})$ be a range space, where $\mathbf{X}$ is a set of points in $\Re^d$ and $\mathcal{R}$ is a set of ranges in $\Re^d$.
For $\eps\in [0,1]$, a subset $C\subseteq \mathbf{X}$ is an $\eps$-sample for $\mathbf{X}$, if for any range $\rect\in \mathcal{R}$, we have
$\bigg|\frac{|\mathbf{X}\cap R|}{|\mathbf{X}|}-\frac{|C\cap R|}{|C|}\bigg|\leq \eps$.
It is known by the $\eps$-sample theorem~\cite{vapnik1971uniform}, that if $(\mathbf{X}, \mathcal{R})$ is a range space with constant VC dimension, then a random subset $\epsample\subseteq \mathbf{X}$ of cardinality $O(\eps^{-2}(\log(\eps^{-1})+\log(\prob^{-1})))$, is an $\eps$-sample for $\mathbf{X}$ with probability at least $1-\prob$.
If $\mathcal{R}$ is a set of rectangles then the random subset $C$ has cardinality $O(\eps^{-2}\log\prob^{-1})$,~\cite{chazelle2000discrepancy, har2011geometric}. 
\new{The next lemma follows straightforwardly by our definitions and properties of $\eps$-samples. Its proof is shown in Appendix~\ref{appndx:prelim}.}
\begin{lemma}
\label{lem:helper1}
Let $\Synop_X$ be the synopsis of $X$ with respect to the class $\F^d_{\square}$. Assume that $\err_{\Synop_X}(\F_{\square}^d)=\delta$ for a parameter $\delta\in[0,1)$.
If $\epsample$ is a set of $O(\eps^{-2}\log \prob^{-1})$ random samples from $\Synop_X$, then $\epsample$ is an $(\eps+\delta)$-sample for $X$
with probability at least $1-\prob^{-1}$. 
\end{lemma}




\paragraph{$\eps$-net}
Let $\sphere^{d-1}$ be the unit sphere in $\Re^d$. A centrally symmetric set $\net\subseteq \sphere^{d-1}$ (i.e., if $u\in \net$, then $-u\in \net$) of $O\left(\eps^{-d+1}\right)$ unit vectors is an \emph{$\eps$-net} if for every point $v\in\sphere^{d-1}$, there exists a point $u\in\net$ with angle at most $\mathsf{cos}^{-1}\left(\frac{1}{1+\eps^2}\right)=O\left(\eps\right)$.
An $\eps$-net in $\Re^d$ can be constructed in $O(\eps^{-d+1})$ time~\cite{agarwal2008robust}.

\paragraph{Range tree}
\new{
Let $X$ be a set of finite weighted points in $\Re^d$. For a point $x\in X$, let $w_x\in \Re$ be its weight, and let $W=\{w_x\mid x\in X\}$ be the list of weights.
The goal is to construct a data structure such that given a query rectangle $\rect$, and an interval $I$ the set $X_{\rect,I}=\{x\in X\mid x\in X\cap \rect, w_x\in I\}$ is reported efficiently.
A range tree $\rangetree$ constructed on $X, W$, is a near-linear space tree-based data structure that constructs a hierarchical decomposition of the space such that, for any rectangle $\rect$ and interval $I$, all points in $X_{\rect,I}$ 
are reported in $O(\polylog(|X|+\outt)$, where $\outt$ is the output size $|X_{\rect,I}|$.
It is known that range trees over weighted points have $O(|X|\log^{d}|X|)$ space, $O(|X|\log^{d}|X|)$ preprocessing time, and can answer range reporting queries in $O(\log^{d}(|X|)+\outt)$ time. See~\cite{de1997computational} for details.
It is also known how to design dynamic range trees that handle insertions and deletions of points, efficiently. More specifically, a dynamic range tree over weighted points has $O(|X|\log^{d}|X|)$ space, $O(|X|\log^{d}|X|)$ preprocessing time, $O(\log^{d+1} |X|)$ update time and can answer range reporting queries in $O(\log^{d+1}(|X|)+\outt)$ time. See~\cite{overmars1983design, chan2018dynamic} for details.
Let $\rangetree=\mathsf{DRangeTreeConstruct}(X,W)$ be a method that takes the set of points $X$ along with their weights $W$ and constructs a dynamic range tree, denoted by $\rangetree$. If we construct a static range tree we use $\mathsf{RangeTreeConstruct}(X,W)$. 
Given a query rectangle $\rect$ and interval $I$, let $\rangetree.\mathsf{Report}(\rect,I)$ be the query procedure that returns $X_{\rect,I}$.
We also define the procedure $\rangetree.\mathsf{ReportFirst}(\rect,I)$ that reports one (arbitrary) point from $X_{\rect,I}$, if $X_{\rect,I}\neq \emptyset$, and then terminate the query procedure. On the other hand, if $X_{\rect,I}= \emptyset$ then it returns $\mathsf{NULL}$. 
By slightly modifying the standard query procedures, it is executed in $O(\log^{d+1}(|X|))$ time and $O(\log^{d}(|X|))$ time for dynamic and static range trees, respectively.
}

\paragraph{Delay guarantees}
A data structure for a reporting (or enumeration) problem over a set $X$ of $n$ points, satisfies $f(n)$ delay (for a real function $f$) if the time between the start of the reporting procedure to the first result, the time between two consecutive pair of results, and the time between the last result and the termination of the reporting procedure are all bounded by $f(n)$ time~\cite{bagan2007acyclic}.

\section{Lower Bounds}
\label{sec:lb}
In this section we show some lower bounds for the \problemDI and \problemDIk problems.
We first show a conditional lower bound for the \problemCDI problem using the strong set intersection conjecture. Next, we show a lower bound for the \problemDIk using a known hardness result for halfspace reporting queries. The missing proofs can be found in Appendix~\ref{appndx:lb}.

\subsection{\problemCDI problem} \label{sec:lower-CDAI}
In this section, we show a conditional lower bound suggesting that no efficient data structure exists for the exact \problemCDI problem. More precisely, assuming the strong set intersection conjecture, we show that for the \problemCDI problem even with one threshold-predicate, there is no data structure that uses near linear ($\O(\totalsize)$) space and answers the queries in near constant ($\O(1)$) time even for $d = 2$. The strong set intersection conjecture is a commonly used tool for showing tradeoffs between the space and the query time of the data structures \cite{krishnan2024range, Wang2023, Zhao2023}, introduced in \cite{Goldstein2017}, and can be stated as follows.

\begin{definition}\label{def:setinter}[Set intersection] Given a universe of integers $\U$, and $g$ subsets $\{S_1, S_2, \cdots, S_g\}$, set intersection query asks for reporting $S_i \cap S_j$ for two arbitrary indices $i, j \in [g]$.
\end{definition}

\begin{conjecture}\label{conj:inter}[Strong set intersection conjecture]
    Given $g$ sets $S_1,\ldots, S_g$ with $M = \sum_{i \in [g]}|S_i|$, any index that answers a set intersection query in time $\O(\qtime + \outt)$, needs at least $\tilde{\Omega}(M+\frac{M^2}{\qtime^2})$ space, \new{where $\qtime$ is any user defined positive integer parameter} and $\outt$ is the size of the output of the query.
\end{conjecture}
We note that a stronger version of Conjecture~\ref{conj:inter}, is stated in \cite{Goldstein2017}, however, this weaker conjecture suffices for us to show the desired lower bounds. 

To show a reduction from the set intersection problem to our \problemCDI problem, we first prove a conditional hardness result for a more restricted version of the set intersection problem, which is of independent interest and could be used to show lower bounds for other restrictive problems. We call a collection of sets \textit{uniform} if any element belongs to the same number of different sets. More formally, a collection of sets $T = \{T_1, \ldots, T_g\}$ over the universe $\U$ is \textit{uniform} if for any pair of elements $u_1, u_2 \in \U$, we have $|\{T_i \in T : u_1 \in T_i\}| = |\{T_i \in T : u_2 \in T_i\}|$. Furthermore, we define the \textit{uniform set intersection} problem similar to \ref{def:setinter}, but with an additional constraint that the input collection of sets is uniform. The proof of the next lemma can be found in Appendix~\ref{appndx:lb}.

\begin{lemma}\label{lem:rest}
    For the uniform set intersection problem with total input size $M$, there is no data structure of size $\O(M)$ that answers queries in time $\O(1 + \outt)$, unless the strong set intersection conjecture is false, where $\outt$ is the size of the output of the query. 
\end{lemma}

In Appendix~\ref{appndx:lb1}, we reduce the uniform set intersection problem to the \problemCDI problem.

\begin{theorem} \label{lowerbound:centralized-percentile}
    Given a repository $\sets=\{\points_1,\ldots, \points_{\setsize}\}$ such that $\totalsize = \sum_{\points_i \in \sets}|\points_i|$ and $\points_i\subset\Re^2$ for every $i\in[\setsize]$, there is no data structure for the \problemCDI problem of size $\O(\totalsize)$ and $\O(1+\outt)$ query time, unless the strong set intersection conjecture is false.
\end{theorem}

\subsection{\problemCDIk problem} \label{sec:lower:CPAI}
In this subsection, we show an unconditional lower bound for the \problemCDIk problem. 
In the halfspace reporting problem, we are given a set $U\subset\Re^d$ of $n$ points and the goal is to construct a data structure such that given a query halfspace $H$, report all points in $|U\cap H|$.
The next lower bound for halfspace reporting queries is known from~\cite{afshani2012improved}. Any data structure that answers $d(d+3)/2$-dimensional halfspace reporting queries in $\bar{\mathcal{Q}}(n)+O(K)$ time, where $K$ is the number of reported points, requires $\Omega\left(\left(\frac{n}{\bar{\mathcal{Q}}(n)}\right)^d2^{-O\left(\sqrt{\log \bar{\mathcal{Q}}(n)}\right)}\right)$ space.
Hence, there is no hope to get an $\O(n)$ space data structure with $\O(1)$ query time for the halfspace reporting problem in $\Re^d$ for $d>4$.
Due to space limit, we show a reduction from the halfspace reporting problem to \problemCDIk problem in Appendix~\ref{appndx:lb2} and we conclude with the next theorem.

\begin{theorem}
\label{thm:lbTopk}
 Given a repository $\sets=\{\points_1,\ldots, \points_{\pointsize}\}$ such that $\totalsize = \sum_{\points_i \in \sets}|\points_i|$ and $\points_i\subset\Re^5$, for every $i\in[\setsize]$, there is no data structure for the \problemCDIk problem of size $\O(\totalsize)$ and $\O(1+\outt)$ query time.
 \end{theorem}

\section{Data structures for the \problemDI Problem}
\label{sec:Perc}
In this section we describe our data structures for the \problemDI problem. In the previous section, we showed that we cannot hope for an efficient exact data structure for the \problemCDI problem when $d\geq 2$, even if the logical expression contains only one threshold-predicate. In Appendix~\ref{subsec:percExact}, we show an efficient exact data structure for the \problemCDI in $\Re^1$ assuming one range-predicate in the logical expression $\Pi$. In this section we focus on approximate data structures for the \problemDI for any $d\geq 1$.
\new{In Subsection~\ref{subsec:techOverview}, we show the high level ideas and a technical overview of our new data structures.}
Then, as a warm-up, in Subsection~\ref{subsec:approxpercopenright}, we consider \problemDI queries with one threshold predicate. In Subsection~\ref{subsec:approxPercRangePred}, we extend the approximate data structure for \problemDI queries with a general range-predicate. Finally, in Appendix~\ref{subsec:percConj}, we further extend to any logical expression (of disjunctions and conjunctions) $\Pi$ over a constant number of range predicates.
\new{Recall that each synopsis $\Synop_{\points}$ supports random sampling (with replacement) over $\points$. Let $\Synop_{\points}.\mathsf{Sample}(\kappa)$ be the synopsis' procedure that returns $\kappa$ random samples from $\points$, for any positive integer $\kappa$.}

\subsection{\new{Overview}}
\label{subsec:techOverview}
\new{
Before we start describing the technical details, we show an overview of our approach. For simplicity we assume that the query logical expression $\Pi$ consists of only one predicate $\pred_{\M_\rect,\interval}$.

Current methods cannot solve \problemDI efficiently with theoretical guarantees, even in the centralized setting where every dataset is known.
For example, consider the following baseline. For every dataset $\points_i$ construct a range tree to answer range counting queries.
Given a query predicate $\pred_{\M_\rect,\interval}$, the naive solution goes through each dataset $\points_i$ and using the range tree it computes $\frac{|\rect\cap \points_i|}{|\points_i|}$. If this value belongs in $\interval$ then we report $i$, otherwise we continue with the next dataset.
While this algorithm solves the $\problemCDI$ problem exactly, it has $\Omega(\setsize)$ query time. Another baseline, that works in both centralized and federated settings, is proposed by~\cite{behme2024fainder} that first defined the $\problemDI$ problem. While their data structure returns a solution for the \problemDI problem with bounded error, the query time is $\Omega(\setsize)$ in the worst case.

We design a new data structure with bounded error having query time with polylogarithmic dependency on $\setsize$. 
Our new method consists of two core ideas: First, we construct a small summary for each dataset using only the input synopses.
Second we construct a data structure over the summaries in a higher dimensional space that allows us to answer percentile queries in sublinear time with respect to $\setsize$.

One of the main components of our data structure is the notion of \emph{coreset}.
Coresets are used in computational geometry, theory and databases, as small summaries that maintain key properties of the original data~\cite{agarwal2008robust, indyk2014composable}.
More specifically, for every dataset $i\in[\setsize]$ we get a set $\epsample_i$ of roughly $O(\eps^{-2})$ random samples from $\Synop_{\points_i}$. Assuming that all synopses have error at most $\delta$, from Lemma~\ref{lem:helper1}, we know that $\epsample_i$ is an $(\eps+\delta)$-sample. By its definition, an $\eps'$-sample acts like a coreset for percentile queries.
Notice that every set $\epsample_i$ can be computed in both settings: In the centralized setting we sample directly from the dataset $\points_i$ and $\delta=0$, so $\epsample_i$ is an $\eps$-sample. In the federated setting we sample from the synopsis $\Synop_{\points_i}$ and $\epsample_i$ is an $(\eps+\delta)$-sample.

While we reduce the size of each dataset $\points_i$ from $n_i$ to $O(\eps^{-2})$ our problem is far from being solved. Given a query rectangle $\rect$ and predicate $\theta$, we still need to go through each set $\epsample_i$ to answer the query.
Our main idea is to construct, all possible combinatorially different rectangles with respect to the coreset $\epsample_i$, for each $i\in[\setsize]$. Interestingly, for each $i\in[\setsize]$, this number is small, because each coreset is small. For each precomputed rectangle $\rec$, we store the weight $\frac{|\rec\cap \epsample_i|}{|\epsample_i|}$, i.e., the fraction of points from the coreset that lies in this rectangle. Intuitively, given a query rectangle $\rect$, we should find for every dataset, among all precomputed rectangles, the one which is the “most similar” to $\rect$ in order to estimate the fraction of points it contains in $\rect$. If the interval $\interval$ has the form $[\intervalL,1]$ (threshold-predicate in Section~\ref{subsec:approxpercopenright}) then it turns out that each dataset $i$ that has at least one precomputed rectangle $\rec$ completely inside $\rect$ with weight at least $\intervalL-\eps-\delta$, should be reported. 
If the interval $\interval$ has the form $[\intervalL,\intervalU]$ (range-predicate in Section~\ref{subsec:approxPercRangePred}), then it turns out that each dataset, whose largest precomputed rectangle inside $\rect$ has weight between $\intervalL-\eps-\delta$ and $\intervalU+\eps+\delta$, should be reported.

In order to guarantee sublinear query time, we construct a geometric data structure among all precomputed rectangles over all datasets. For threshold-predicate (Section~\ref{subsec:approxpercopenright}), we search among all precomputed rectangles and find any of them that lies completely inside $\rect$ with weight at least $\intervalL\!-\!\eps\!-\!\delta$ (property $\textbf{P}1$). Similarly, for range-predicate (Section~\ref{subsec:approxPercRangePred}) we search among all precomputed rectangles to find any of them that is maximal inside $\rect$ (no other larger precomputed rectangle generated from the same coreset lies completely inside $\rect$) with weight between $\intervalL\!-\!\eps\!-\!\delta$ and $\intervalU\!+\!\eps\!+\!\delta$ (property $\textbf{P}2$).
Once we find such a rectangle (that was generated from coreset, say, $\epsample_i$), we report $i$, we remove all precomputed rectangles generated from $\epsample_i$ and continue with the same query until we cannot find a rectangle that satisfies property $\textbf{P}1$ (threshold-predicate), or property $\textbf{P}2$ (range-predicate).
In order to efficiently find precomputed rectangles with these properties and to efficiently remove precomputed rectangles, we use a dynamic range tree.

Recall that the range tree is constructed on a set of points, such that given a query rectangle it reports the points inside the query rectangle efficiently. In order to use it for our purpose, in Section~\ref{subsec:approxpercopenright}, we map the precomputed rectangles in $\Re^d$ (along with their weights) to weighted points in $\Re^{2d}$.
The new set of (weighted) points in $\Re^{2d}$ is stored in the range tree. Given a query rectangle $\rect$ in $\Re^d$, we map it to a new rectangle (orthant) $\rect'$ in $\Re^{2d}$, such that a precomputed rectangle $\rec$ in $\Re^d$ lies completely inside $\rect$, if and only if, $\rec$'s mapped point in $\Re^{2d}$ lies inside the orthant $\rect'$.
In Section~\ref{subsec:approxPercRangePred} the geometric data structure is more involved and we need to map the precomputed rectangles to points in $\Re^{4d}$.
We show the details in Section~\ref{subsec:approxPercRangePred}.
}

\begin{figure}[h]
\vspace{-1.8em}
    \centering
    \subfloat[\new{Sets $\epsample_1, \epsample_2$ and the precomputed intervals $\mathcal{R}_1, \mathcal{R}_2$.}]{\includegraphics[width=0.3\textwidth]{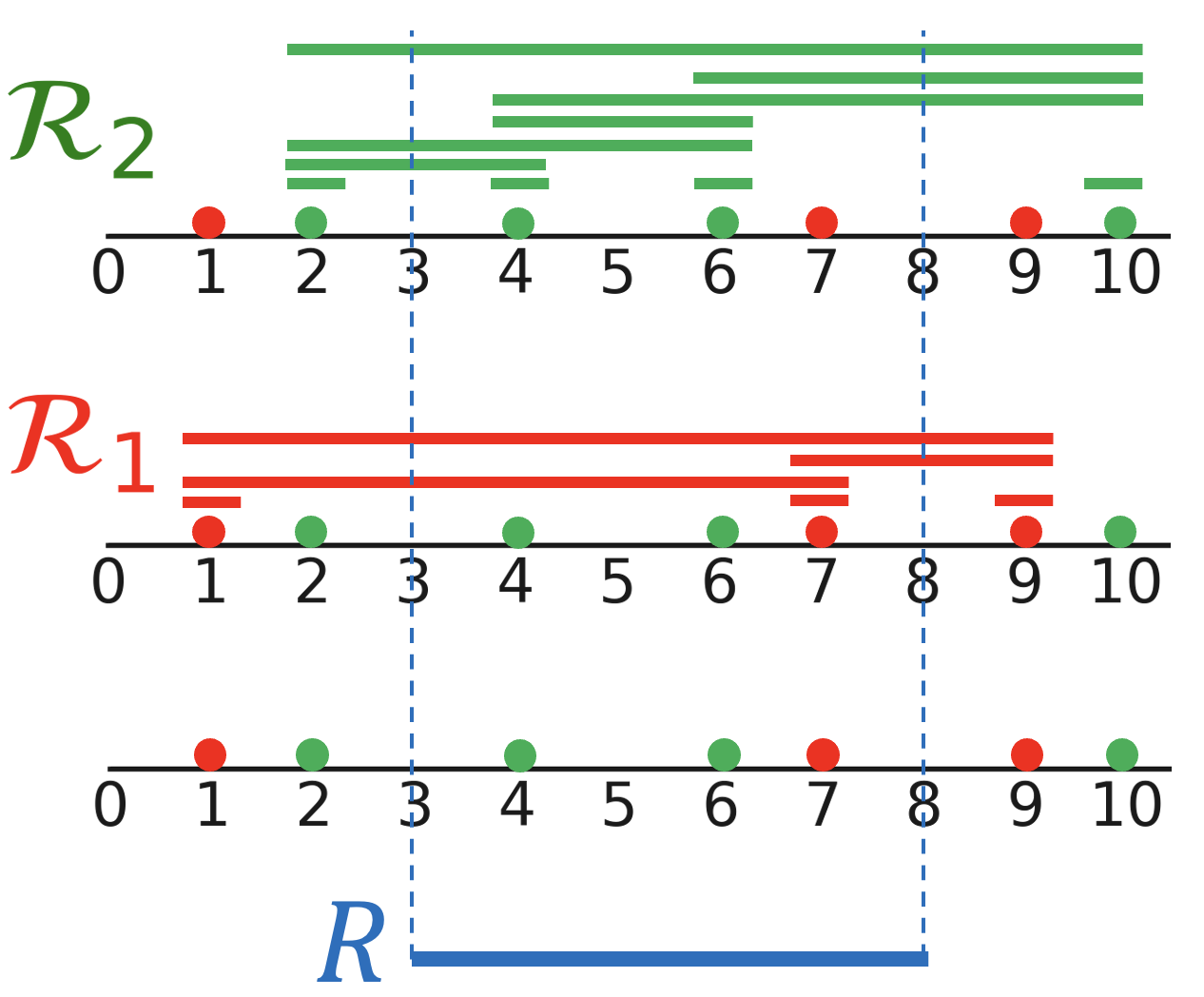} \label{fig:figEx1}}
    \subfloat[\new{Set of weighted points $Q$}]{\includegraphics[width=0.3\textwidth]{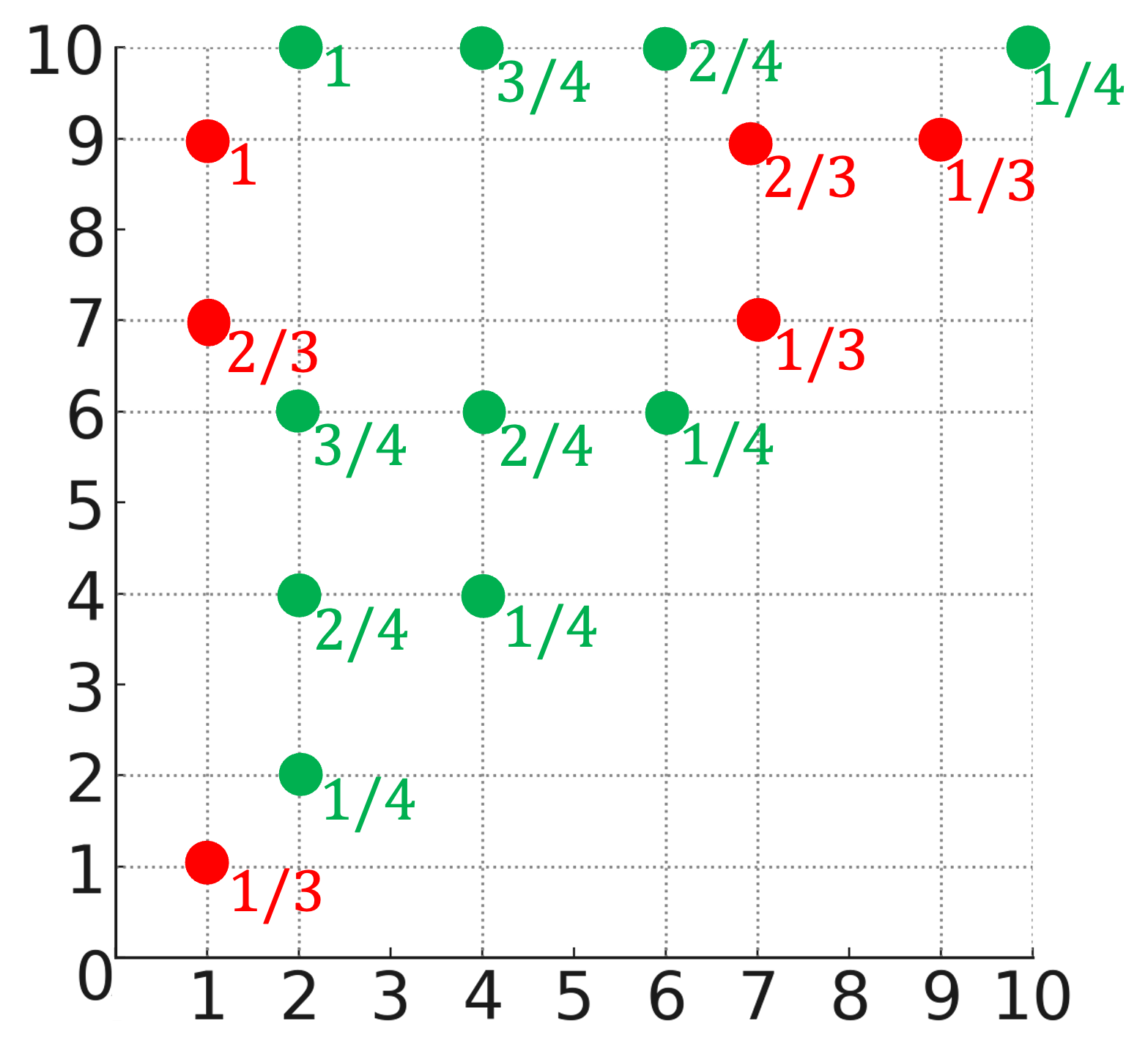} \label{fig:figEx2}}
    \subfloat[\new{Set of weighted points $Q$ along with the orthant $\rect'$.}]{\includegraphics[width=0.3\textwidth]{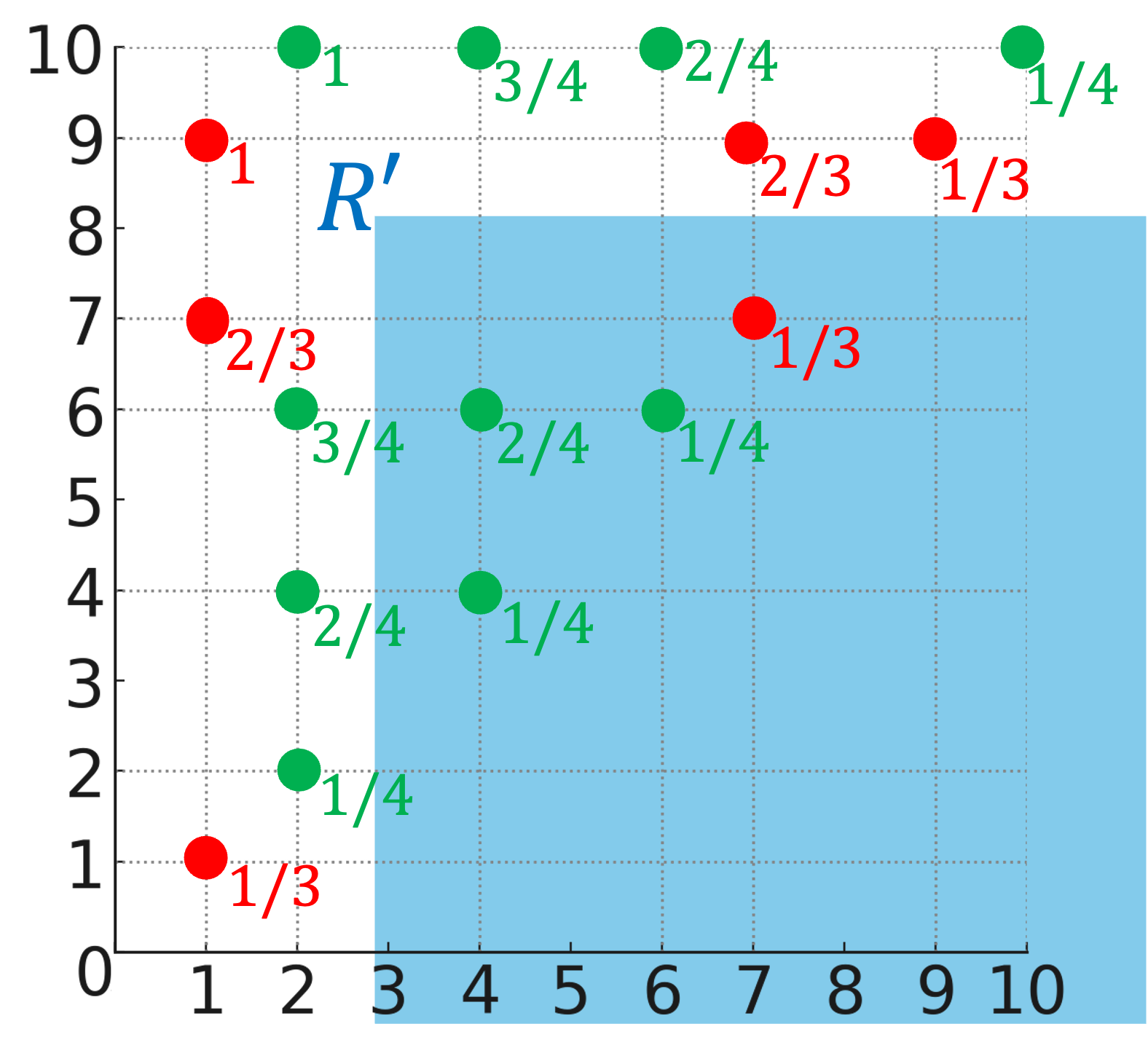} \label{fig:figEx3}}
    \vspace{-0.5em}
    \caption{\new{Toy example.}}
    \label{fig:all}
    \vspace{-1.8em}
\end{figure}

\subsection{Approximate data structure for the \problemDI problem with a threshold-predicate}
\label{subsec:approxpercopenright}
We first design our data structure and prove its guarantees assuming that for every $i\in[\setsize]$, $\err_{\Synop_{\DS_i}}(\F^d_{\square}) \leq\delta_i\leq \delta$, for a known parameter $\delta\in [0,1)$. In the end, we give the guarantees assuming that $\delta_i$'s are unknown.
If $\delta=0$, then this is the \problemCDI problem, otherwise it is the \problemFDI problem.
In this subsection, we focus on threshold-predicates, i.e., the query interval $\interval$ is defined as $[\intervalL,1]$.


\begin{algorithm}[t]
  \caption{\new{$\mathsf{ConstructT}(\Synop_{\points_1}, \ldots, \Synop_{\points_\setsize},\eps,\prob)$}}
  \label{alg:construction}
  \small
  $Q=\emptyset$, $W=\emptyset$\;
  \ForEach{$i\in[\setsize]$}{ 
    $Q_i=\emptyset$, $W_i=\emptyset$\;
    $\epsample_i=\Synop_{\points_i}.\mathsf{Sample}( \Theta(\eps^{-2}\log(\setsize\prob^{-1})))$\;
    $\mathcal{R}_i=\text{ All combinatorially different hyper-rectangles defined by }        
        \epsample_i$\;
    \ForEach{$\rec\in \mathcal{R}_i$}{
        $q_{\rec}=(\rec^-_1,\ldots, \rec^+_d, \rec^+_1,\ldots, \rec^+_d)$, \hspace{0.2em} $q_{\rec}.\mathsf{IndexDataset}=i$, \hspace{0.2em}
        $w_{q_\rec}=\frac{|\rec\cap \epsample_i|}{|\epsample_i|}$\;
        $Q_i=Q_i\cup\{q_{\rec}\}$, $W_i=W_i\cup\{w_{q_\rec}\}$\;
    }
    $Q=Q\cup Q_i$, $W=W\cup W_i$\;
  }
  $\rangetree=\mathsf{DRangeTreeConstruct}(Q,W)$\;
  \Return $\rangetree$\;
\end{algorithm}
\begin{algorithm}[t]
  \caption{\new{$\mathsf{QueryT}(\rangetree, Q, \rect, \interval, \eps,\delta)$}}
  \label{alg:query}
  \small
  $\rect'=[\rect^-_1,\infty)\times \ldots\times [\rect^-_d,\infty)\times (-\infty,\rect^+_1]\times\ldots\times (-\infty,\rect^+_d]$, \hspace{0.2em}
  $I'=[\intervalL-\eps-\delta,1]$\;
  $q_\rec=\rangetree.\mathsf{ReportFirst}
  (\rect', I')$,\hspace{0.2em} $J=\emptyset$\;
  \While{$q_\rec\neq \mathsf{NULL}$}{
    $j=q_\rec.\mathsf{IndexDataset}$, $J=J\cup\{j\}$,\hspace{0.2em}
    $\mathsf{Report}\hspace{0.3em} j$\;
    \lForEach{$q\in Q_j$}{Delete $q$ from $\rangetree$}
    $q_\rec=\rangetree.\mathsf{ReportFirst}
  (\rect', I')$\;
   }
\lForEach{$j\in J$}{Re-insert all points from $Q_j$ in $\rangetree$}
  \Return\;
\end{algorithm}

\paragraph{Data structure}
\new{In Algorithm~\ref{alg:construction}, we show the pseudocode of the construction phase.}
For every $i\in[\setsize]$, we sample a set $\epsample_i$ of $\Theta(\eps^{-2}\log(\setsize\prob^{-1}))$ points from $\Synop_{\points_i}$.
From Lemma~\ref{lem:helper1}, we know that $\epsample_i$ is a $(\eps+\delta)$-sample with probability at least $\prob^{-1}$.
For every $i\in [\setsize]$, we construct the set $\mathcal{R}_i$ that contains all combinatorially different hyper-rectangles defined by $\epsample_i$. Every hyper-rectangle $\rec\in \mathcal{R}_i$ is defined by their two opposite corners $\rec^-,\rec^+\in \Re^d$ such that $\rec^-_h\leq \rec^+_h$, for every $h\in[d]$, where $\rec^-_h, \rec^+_h$ are the $h$-coordinates of $\rec^-$ and $\rec^+$, respectively.
Let $q_\rec=(\rec^-_1,\ldots, \rec^-_d, \rec^+_1,\ldots, \rec^+_d)$ be the point in $\Re^{2d}$ defined by merging the opposite corners of $\rec$.
We also store in $q_\rec.\mathsf{IndexDataset}$ the index $i$ so that we know the index of the dataset/synopsis that $q_\rec$ was constructed for (it will be useful in the query phase).
We set the weight of $q_\rec$ to be $w_{q_\rec}=\M_\rec(\epsample_i)=\frac{|\rec\cap \epsample_i|}{|\epsample_i|}$.
Let $Q_i=\{q_\rec\mid \rec\in \mathcal{R}_i\}$ and $W_i=\{w_{q_\rec}\mid q_\rec\in Q_i\}$. Finally, we define $Q=\bigcup_{i\in [\setsize]}Q_i$ and $W=\bigcup_{i\in [\setsize]}W_i$. We construct a dynamic range tree $\rangetree$ over the (weighted) points $Q$ with weights $W$ calling the function $\mathsf{DRangeTreeConstruct}(Q,W)$.

\paragraph{\textit{Example}}
\new{
As an example, we show the construction of our data structure for $d=1$. We consider two datasets, $\points_1$ and $\points_2$. In Figure~\ref{fig:figEx1} we show the sampled points $\epsample_1=\{1, 7, 9\}$ and $\epsample_2=\{2, 4, 6,10\}$. We construct the sets $\mathcal{R}_1$ and $\mathcal{R}_2$ as shown in Figure~\ref{fig:figEx1} containing all combinatorially different rectangles (intervals for $d=1$) defined by $\epsample_1$ and $\epsample_2$, respectively. For example, $\mathcal{R}_1$ contains the intervals $[1,1], [7,7], [9,9], [1,7], [1,9], [7,9]$ and $\mathcal{R}_2$ contains the intervals $[2,2], [4,4], [6,6], [10,10], [2,4], [2,6], [2,10], [4,6], [4,10], [6,10]$. Then for each $\rec\in \mathcal{R}_1$ (or $\mathcal{R}_2$) we construct the point $q_{\rec}\in \Re^2$ as shown in Figure~\ref{fig:figEx2}. For example, for the interval $\rec=[1,7]$, we construct the point $q_{\rec}=(1,7)$. The weight of point $q_{\rec}=(1,7)$ is $2/3$ because $w_{q_{\rec}}=\frac{|\rec\cap \epsample_1|}{|\epsample_1|}=\frac{2}{3}$. In the end, the dynamic range tree $\rangetree$ is constructed over the weighted points in Figure~\ref{fig:figEx2}.
}

\paragraph{Query procedure}
\new{In Algorithm~\ref{alg:query}, we show the pseudocode of the query phase.
We are given a query rectangle $\rect$ in $\Re^d$ and a one-sided query interval $\interval=[\intervalL,1]$. Let $\rect^-, \rect^+$ be the two opposite corners of $\rect$. We define the orthant $\rect'=[\rect^-_1,\infty)\times \ldots\times [\rect^-_d,\infty)\times (-\infty,\rect^+_1]\times\ldots\times (-\infty,\rect^+_d]$ in $\Re^{2d}$ and the weight query threshold interval $I'=[\intervalL-\eps-\delta,1]$.
We repeat the following until no point is returned:
Using the range tree's $\rangetree$ query procedure $\rangetree.\mathsf{ReportFirst}(\rect',I')$ we compute, if any, a point $q_\rec$ in $Q\cap \rect'$ with weight in $I'$.
Let $j$ be the index of the synopsis $\Synop_{\points_j}$ that $q_\rec$ was constructed for. We report $j$, and we delete from $\rangetree$ all points in $Q_j$. We continue with the next iteration.
If such point $q_\rec$ is not found, we stop the execution.
Let $J$ be the indexes we reported.
In the end of the query-process, we re-insert in $\rangetree$ all points $Q_j$ for every index $j\in J$.}

\paragraph{Example (cont.)}
\new{
We show the query procedure executed on our running example from Figure~\ref{fig:all}. Without loss of generality, assume $
\delta=0$ (centralized setting) and $\eps<0.01$.\footnote{We note that if $\eps<0.01$, then we would need more sampled points from each dataset. However, this assumption simplifies the description of the example and successfully derives the main idea of our approach. Another way to think of this is that there exist more samples very far from coordinates $[1,10]$ that do not play any role in the query $\rect=[3,8]$.}
Let $\rect=[3,8]$ be the query rectangle (interval for $d=1$), and let $\interval=[0.2,1]$. We map $\rect$ to the orthant $\rect'=[3,\infty)\times (-\infty,8]$ in $\Re^2$ as shown in the blue area of Figure~\ref{fig:figEx3}. We also have $I'=[0.2-\eps,1]$. 
Notice that all points in the blue area have weight at least $0.2$.
Assume that the procedure $\rangetree.\mathsf{ReportFirst}(\rect', I')$ first returns the point $(4,4)$. We report the index $2$ (because $(4,4)$ was constructed from $\epsample_2$). Then we remove from $\rangetree$ all green points. In the next iteration, the procedure $\rangetree.\mathsf{ReportFirst} (\rect', I')$ returns the point $(7,7)$ because it lies in the blue orthant and its weight is greater than $0.2$. The index $1$ is reported and all red points are removed from $\rangetree$. In the end all removed points are re-inserted in $\rangetree$.
Both indexes $1$ and $2$ are correctly reported because $\frac{|\rect\cap\points_1|}{|\points_1|}\geq \frac{|\rect\cap\epsample_1|}{|\epsample_1|}-\eps\geq w_{(4,4)}-\eps> 0.2$ and $\frac{|\rect\cap\points_2|}{|\points_2|}\geq \frac{|\rect\cap\epsample_2|}{|\epsample_2|}-\eps\geq w_{(7,7)}-\eps> 0.2$.
}

\paragraph{Correctness}
Let $\Pi=\pred_{\M_{\rect},\interval}$ be the logical expression with one threshold-predicate.
We first show that we report all the ``correct'' indexes $\query_\Pi(\sets)$.
\begin{lemma}
    $\query_\Pi(\sets)\subseteq J$, with probability at least $1-\prob$.
\end{lemma}
\begin{proof}
From Lemma~\ref{lem:helper1} we have that every $\epsample_i$ is an $(\eps+\delta)$-sample for $\points_i$ with probability at least $1-\prob$. Let $i\in \query_\Pi(\sets)$ be an index such that $\intervalL\leq \M_\rect(\points_i)$. We show that $i\in J$. Let $\rec$ be the largest rectangle in $\mathcal{R}_i$ such that $\rec\subseteq \rect$. 
  By definition, it holds that $\rec\cap \epsample_i = \rect\cap \epsample_i$. If we show that $q_\rec\in \rect'\cap Q$ and $w_{q_\rec}\geq \intervalL-\eps-\delta$, then the result follows. Indeed, by definition of $\epsample_i$, we have $\intervalL-\eps-\delta\leq \frac{|\points_i\cap \rect|}{|\points_i|}-\eps-\delta\leq \frac{|\epsample_i\cap \rect|}{|\epsample_i|}=\frac{|\epsample_i\cap \rec|}{|\epsample_i|}=\M_\rec(\epsample_i)=w_{q_\rec}$. Furthermore, since $\rec\subseteq \rect$, we have $\rect^-_h\leq \rec^-_h$ and $\rect^+_h\geq \rec^+_h$ for every $h\in [d]$. Hence, by the definition of the orthant $\rect'$, it holds that $q_\rec\in \rect'\cap Q$.
\end{proof}

\begin{lemma}
    For every $j\in J$ it holds that $ \M_\rect(\points_j)\geq \intervalL-2\eps-2\delta$, with probability at least $1-\prob$. 
\end{lemma}
\begin{proof}
    Let $j\in J$ be an index reported by our query procedure. Let $q_\rec\in Q_j$ be the point found in $\rect'$ at the moment that we reported $j$. By definition, $w_{q_\rec}=\M_\rec(\epsample_i)=\frac{|\rec\cap \epsample_j|}{|\epsample_j|}\geq \intervalL-\eps-\delta$, and  $\rec\subseteq \rect$. Hence, we have $|\rec\cap \epsample_j|\leq |\rect\cap \epsample_j|$. By definition of $\epsample_j$, we have $\M_\rect(\points_j)=\frac{|\points_j\cap \rect|}{|\points_j|}\geq \frac{|\rect\cap \epsample_j|}{|\epsample_j|}-\eps-\delta\geq \frac{|\rec\cap \epsample_j|}{|\epsample_j|}-\eps-\delta\geq \intervalL-2\eps-2\delta$.
\end{proof}


\paragraph{Analysis}
\new{For simplicity, we assume that $\eps$ is a an arbitrarily small constant and $\prob=\frac{1}{\setsize}$. In Appendix~\ref{appndx:approxpercopenright} we show the analysis for arbitrary values of $\eps$ and $\prob$.}
The proof of the next lemma can be found in Appendix~\ref{appndx:approxpercopenright}. We set $\eps\leftarrow \eps/2$ and we get the final result.
\begin{lemma}
\label{lem:space1}
\new{
The data structure has $O(\setsize\cdot\log^{4d}(\setsize))$ space and it can be constructed in $O(\setsize\cdot\log^{4d}(\setsize))$ time.
The query procedure runs in $O(\log^{2d+1}(\setsize)+\out\cdot\log^{4d+1}(\setsize))$ time.} 
\end{lemma}

\begin{theorem}
\label{thm:res1}
Let $\{\Synop_{\points_1},\ldots,\Synop_{\points_\setsize}\}$ be the input to the $\problemDI$ problem, such that $\Synop_{\points_i}$ is a synopsis of a dataset $\points_i\subset \Re^d$, where $d\geq 1$ is a constant, with $\err_{\Synop_{\points_i}}(\F^d_{\square})\leq \delta$ for $\delta\in[0,1)$, for every $i\in[\setsize]$. Let $\eps\in (0,1)$ be an arbitrarily small constant parameter.
A data structure of size $O(\setsize\cdot\log^{4d}(\setsize))$ can be constructed in $O(\setsize\cdot\log^{4d}(\setsize))$  time, such that
given a query $\Pi=\pred_{\M_\rect,\interval}$ for a one-sided interval $\interval=[\intervalL,1]$,
it returns a set of indexes $J$ such that $\query_\Pi(\sets)\subseteq J$ and for every $j\in J$, $\M_\rect(\points_j)\geq \intervalL-\eps-2\delta$, with probability at least $1-\frac{1}{\setsize}$.
The query time is $O(\log^{2d+1}(\setsize)+\out\cdot\log^{4d+1}(\setsize))$.

\end{theorem}

\subsection{Approximate data structure for the \problemDI problem with a range-predicate}
\label{subsec:approxPercRangePred}

\begin{figure}[t]
 \begin{minipage}{0.4\linewidth}
        \centering
        \includegraphics[scale=0.3]{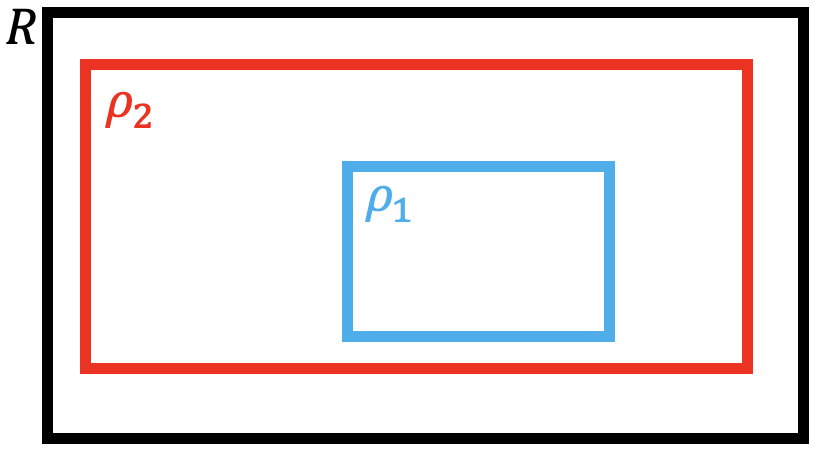}
    \caption{Let $\eps=0.01$,  $\delta=0$, and $\rect$ is the query rectangle. Let $w_{q_{\rec_1}}=\frac{|\epsample_j\cap \rec_1|}{|\epsample|}=0.1$ and $w_{q_{\rec_2}}=\frac{|\epsample_j\cap \rec_2|}{|\epsample|}=0.9$, while $\interval=[0, 0.2]$. If the query procedure first finds $q_{\rec_1}$ it will report the index $j$. However, the actual measure function of $\points_j$ is at least $0.9-0.01\gg 0.2$. Hence $j$ should not be reported.} 
    \label{fig:fig1} 
\end{minipage}
\hspace{1cm}
\begin{minipage}{0.4\linewidth}
        \centering
        \vspace{-2mm}
        \includegraphics[scale=0.3]{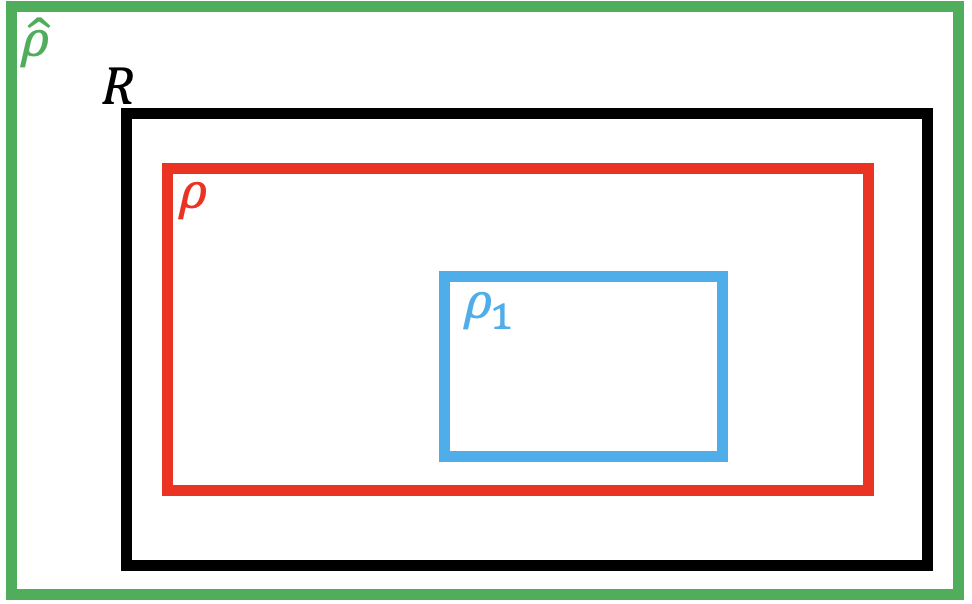}
    \caption{If $\rect$ is the query rectangle and $\rec$ is the largest rectangle in $\rect$ then $\rec\subseteq \rect\subset\!\!\subset\hat{\rec}$. On the other hand, for any other rectangle $\rec_1\in \mathcal{R}_j$
    with $\rec_1\subset\rec$
    there is no point $q_{(\rec_1,\hat{\rec}_1)}\in Q$ such that $\rect\subset\!\!\subset \hat{\rec}_1$ because there exists $\rec$ with $\rec_1\subset \rec\subseteq \rect\subset\!\!\subset\hat{\rec}_1$.} 
    \label{fig:fig2} 
\end{minipage}
\vspace{-1em}
\end{figure}
\new{
Unfortunately, the data structure in the previous subsection, does not extend to range-predicates. 
Previously, any arbitrary precomputed rectangle (with sufficient weight) completely inside $\rect$ could be selected. However this is not valid for range predicates. For example check Figure~\ref{fig:fig1}. 
Similarly, in the example from Figure~\ref{fig:all}, assume that $\rect=[3,8]$, $\interval=[0.2, 0.4]$, and $\eps<0.01$.
Consider the precomputed interval $[4,4]$ in $\mathcal{R}_2$. This interval lies inside $\rect$ and it has weight $1/4\in [0.2-\eps,0.4+\eps]$.
So it is returned by $\rangetree.\mathsf{ReportFirst}(\rect',I')$.
However, $\frac{|\rect\cap \points_2|}{|\points_2|}\geq \frac{|\rect\cap \epsample_2|}{|\epsample_2|}-\eps=w_{(4,6)}-\eps=2/4-\eps> 0.4$, so the dataset with index $2$ should not be reported.

As we highlighted at the beginning of this section, in order to resolve the issue, we should always consider maximal precomputed rectangles in the data structure. A precomputed rectangle $\rec$ (generated from $\epsample_j$) is maximal with respect to $\rect$ if $\rec$ 
lies inside $\rect$ and there is no other precomputed rectangle (generated from the same $\epsample_j$) in $\rect$ that fully contains $\rec$. In the example from Figure~\ref{fig:all}, $\epsample_1$'s maximal precomputed rectangle (interval in $d=1$) in $\rect$ is $[7,7]$, while $\epsample_2$'s maximal precomputed interval in $\rect$ is $[4,6]$. Hence, only the points $(7,7)$ and $(4,6)$ should be considered to answer the query. Indeed, for $\interval=[0.2,0.4]$, the weight of point $(4,6)$ is sufficiently larger than $0.4$, so index $2$ (correctly) is not reported. On the other hand, the index $1$ is reported because the weight of point $(7,7)$ (that corresponds to the maximal interval $[7,7]$) lies in $[0.2-\eps,0.4+\eps]$.

In order to search over maximal rectangles,
we compute and store all pairs of precomputed rectangles $\rec,\hat{\rec}$ (generated from the same coreset) such that $\rec\subseteq \hat{\rec}$ and there is no other precomputed rectangle $\rec'$ such that $\rec\subset \rec'\subset\!\!\subset \hat{\rec}$, where $\rec'\subset\!\!\subset\hat{\rec}$ denotes that $\rec'\subset \hat{\rec}$ and the boundary of $\rec'$ does not intersect the boundary of $\hat{\rec}$.
Furthermore, we assign the weight $\frac{|\rec\cap \epsample_i|}{|\epsample_i|}$ to the pair $\rec,\hat{\rec}$.
We observe that for any query rectangle $\rect$, only its maximal precomputed rectangle $\rec$ will belong to at least one pair $(\rec,\hat{\rec})$ such that $\rec\subseteq \rect\subset\!\!\subset \hat{\rec}$. For example, see Figure~\ref{fig:fig2}.

From the discussion above, our goal is to construct a data structure such that, given a query rectangle $\rect$, and an interval $\interval$, it returns  (if it exists) a pair of precomputed rectangles $\rec,\hat{\rec}$ such that $\rec\subseteq\rect\subset\!\!\subset \hat{\rec}$ with weight between $\intervalL-\eps-\delta$ and $\intervalU+\eps+\delta$. Then, remove all pairs of rectangles that contain $\rec$ and continue until no pair is returned. We use a dynamic range tree to find such pairs. Since a range tree stores a set of points, we map every stored pair of precomputed rectangles (along with its weight) to a weighted point in $\Re^{4d}$. Given a query rectangle $\rect$ in $\Re^d$, we map it to an orthant $\rect'$ in $\Re^{4d}$ such that there exists a stored pair $\rec,\hat{\rec}$ with $\rec\subseteq \rect\subset\!\!\subset \hat{\rec}$, if and only if, pair's mapped point in $\Re^{4d}$ lies inside the orthant $\rect'$.
}

\begin{algorithm}[t]
  \caption{\new{$\mathsf{ConstructR}(\Synop_{\points_1}, \ldots, \Synop_{\points_\setsize},\eps,\prob)$}}
  \label{alg:Gconstruction}
   \small
  $Q=\emptyset$, $W=\emptyset$, $\mathcal{B}$= bounding box\;
  \ForEach{$i\in[\setsize]$}{ 
    $Q_i=\emptyset$, $W_i=\emptyset$\;
    $\epsample_i=\Synop_{\points_i}.\mathsf{Sample}( \Theta(\eps^{-2}\log(\setsize\prob^{-1})))$\;
    $\bar{\epsample}_i=\text{ Projections of } \epsample_i \text{ onto the } 2\cdot d \text{ boundaries of } \mathcal{B}$\;
    $\mathcal{R}_i=\text{ All combinatorially different hyper-rectangles defined by }        
        \epsample_i\cup \bar{\epsample}_i$, \hspace{0.2em}
    $\hat{\mathcal{R}}_i=\emptyset$\;
    \ForEach{pair $(\rec,\hat{\rec})\in \mathcal{R}_i\times \mathcal{R}_i$ with $\rec\subseteq \hat{\rec}$}{
        \lIf{$\nexists \rec'\in \mathcal{R}_i: \rec\subset \rec'\subset\!\!\subset \hat{\rec}$}{
            $\hat{\mathcal{R}}_i=\hat{\mathcal{R}}_i\cup \{(\rec,\hat{\rec})\}$
        }
    }
    
    \ForEach{$(\rec,\hat{\rec})\in \hat{\mathcal{R}}_i$}{
        $q_{(\rec,\hat{\rec})}=(\rec^-_1,\ldots, \rec^-_d,\hat{\rec}^-_1,\ldots, \hat{\rec}^-_d, \rec^+_1,\ldots, \rec^+_d, \hat{\rec}^+_1,\ldots, \hat{\rec}^+_d)$, \hspace{0.2em} $q_{(\rec,\hat{\rec})}.\mathsf{IndexDataset}=i$\;
        $w_{q_{(\rec,\hat{\rec})}}=\frac{|\rec\cap \epsample_i|}{|\epsample_i|}$\;
        $Q_i=Q_i\cup\{q_{(\rec,\hat{\rec})}\}$, $W_i=W_i\cup\{w_{q_{(\rec,\hat{\rec})}}\}$\;
    }
    $Q=Q\cup Q_i$, $W=W\cup W_i$\;
  }
  $\rangetree=\mathsf{DRangeTreeConstruct}(Q,W)$\;
  \Return $\rangetree$\;
\end{algorithm}

\begin{algorithm}[t]
  \caption{\new{$\mathsf{QueryR}(\rangetree, Q, \rect, \interval, \eps,\delta)$}}
  \label{alg:Gquery}
  \small
  $\rect'\!\!=\!\![\rect^-_1,\infty)\!\times\! \ldots\!\times\! [\rect^-_d,\infty)
\!\times\! (-\infty,\rect^-_1)\!\times\!\ldots\times (-\infty,\rect^-_d)\times (-\infty,\rect^+_1]\times\ldots\times(-\infty,\rect^+_d]\times (\rect^+_1,\infty)\times\ldots\times (\rect^+_d,\infty)$\;
$I'=[\intervalL-\eps-\delta,\intervalU+\eps+\delta]$\;
  $q_{(\rec,\hat{\rec})}=\rangetree.\mathsf{ReportFirst}
  (\rect', I')$,\hspace{0.2em}
  $J=\emptyset$\;
  \While{$q_{(\rec,\hat{\rec})}\neq \mathsf{NULL}$}{
    $j=q_{(\rec,\hat{\rec})}.\mathsf{IndexDataset}$, $J=J\cup\{j\}$,\hspace{0.2em}
    $\mathsf{Report}\hspace{0.3em} j$\;
    \lForEach{$q\in Q_j$}{Delete $q$ from $\rangetree$}
    $q_{(\rec,\hat{\rec})}=\rangetree.\mathsf{ReportFirst}
  (\rect', I')$\;
   }
\lForEach{$j\in J$}{Re-insert all points from $Q_j$ in $\rangetree$}
  \Return\;
\end{algorithm}

\paragraph{Data structure}
\new{In Algorithm~\ref{alg:Gconstruction}, we show the pseudocode of the construction phase.}
Without loss of generality assume that all datasets lie in a bounded box $\mathcal{B}$.
As we had in Subsection~\ref{subsec:approxpercopenright}, 
for every $i\in[\setsize]$, we get an $(\eps+\delta)$-sample $\epsample_i\subset\Re^d$ by sampling $\Theta(\eps^{-2}\log(\setsize\prob^{-1}))$ points from $\Synop_{\points_i}$ uniformly at random.
For every sample point $s\in\epsample$ we get its projections on the $2\cdot d$ facets of $\mathcal{B}$, and let $\bar{\epsample}_i$ be the set of projections.\footnote{The bounding box $\mathcal{B}$ and the projections $\bigcup_{i\in[\setsize]}\bar{\epsample}_i$ are needed only in this subsection for technical reasons in the correctness proof.}
For every $i\in [\setsize]$, we construct the set $\mathcal{R}_i$ that contains all combinatorially different hyper-rectangles defined by $\epsample_i\cup\bar{\epsample}_i$. Every hyper-rectangle $\rec\in \mathcal{R}_i$ is defined by their two opposite corners $\rec^-,\rec^+\in \Re^d$ such that $\rec^-_h\leq \rec^+_h$, for every $h\in[d]$, where $\rec^-_h, \rec^+_h$ are the $h$-coordinates of $\rec^-$ and $\rec^+$, respectively.
We construct all pairs of rectangles $\hat{\mathcal{R}}_i=\{(\rec, \hat{\rec})\mid \rec\in \mathcal{R}_i, \hat{\rec}\in \mathcal{R}_i, \rec\subseteq \hat{\rec}, \nexists \rec'\in \mathcal{R}_i: \rec\subset \rec'\subset\!\!\subset \hat{\rec}\}$. For every pair $(\rec,\hat{\rec})\in \hat{\mathcal{R}}$, let 
$q_{(\rec,\hat{\rec})}=(\rec^-_1,\ldots, \rec^-_d,\hat{\rec}^-_1,\ldots, \hat{\rec}^-_d, \rec^+_1,\ldots, \rec^+_d, \hat{\rec}^+_1,\ldots, \hat{\rec}^+_d)$ be the point in $\Re^{4d}$ defined by merging the opposite corners of $\rec$ and $\hat{\rec}$.
We also store in $q_{(\rec,\hat{\rec})}.\mathsf{IndexDataset}$ the index $i$ so that we know the index of the dataset/synopsis that $q_{(\rec,\hat{\rec})}$ was constructed for.
We set the weight of $q_{(\rec,\hat{\rec})}$ to be $w_{q_{(\rec,\hat{\rec})}}=\M_\rec(\epsample_i)=\frac{|\rec\cap \epsample_i|}{|\epsample_i|}$. Notice that the weight $w_{q_{(\rec,\hat{\rec})}}$ is defined only with respect to $\rec$.
Let $Q_i=\{q_{(\rec,\hat{\rec})}\mid (\rec,\hat{\rec})\in \hat{\mathcal{R}}_i\}$ and $W_i=\{w_{q_{(\rec,\hat{\rec})}}\mid q_{(\rec,\hat{\rec})}\in Q_i\}$. Finally, we define $Q=\bigcup_{i\in [\setsize]}Q_i$ and $W=\bigcup_{i\in [\setsize]}W_i$. We construct a dynamic range tree $\rangetree$ over the (weighted) points $Q$ with weights $W$ calling the function $\mathsf{DRangeTreeConstruct}(Q,W)$.

\paragraph{\textit{Example}}
\new{
Back to our example from Figure~\ref{fig:all}.
We construct all pairs in $\hat{\mathcal{R}}_1$ and $\hat{\mathcal{R}}_2$.
For instance, the pair $([7,7], [1,9])$ belongs in $\hat{\mathcal{R}}_1$ because $[7,7]$ is a subset of $[1,9]$ and there is no other interval in $\mathcal{R}_1$ strictly inside $[1,9]$ that contains $[7,7]$ (while $[7,7]\subset [7,9]\subset [1,9]$ it does not hold that $[7,9]\subset\!\!\subset [1,9]$).
Similarly, the pair $([4,6], [2,10])$ belongs in $\hat{\mathcal{R}}_2$. On the other hand the pair $([6,6],[2,10])\notin \hat{\mathcal{R}}_2$ because $[6,6]\subset [4,6]\subset\!\!\subset [2,10]$.
Each pair in $\hat{\mathcal{R}}_1\cup \hat{\mathcal{R}}_2$ is mapped to a weighted point in $\Re^4$, for example the pair $([7,7], [1,9])$ is mapped to the point $(7,1,7,9)\in \Re^4$ with weight $\frac{[7,7]\cap \epsample_1}{|\epsample_1|}=\frac{1}{3}$. The dynamic range tree is constructed over the weighted points in $\Re^4$.
}

\paragraph{Query procedure}
\new{In Algorithm~\ref{alg:Gquery}, we show the pseudocode of the query phase.
We are given a query rectangle $\rect$ in $\Re^d$ and a query interval $\interval=[\intervalL,\intervalU]$.
Without loss of generality, assume that the query rectangle $\rect$ is contained in the bounding box $\mathcal{B}$.
Let $\rect^-, \rect^+$ be the two opposite corners of $\rect$.
Without loss of generality we can assume that there is no point from $Q$ that lies on an axis-align hyperplane that is tangent to $\rect$.
We define the orthant $\rect'=[\rect^-_1,\infty)\times \ldots\times [\rect^-_d,\infty)
\times (-\infty,\rect^-_1)\times\ldots\times (-\infty,\rect^-_d)\times (-\infty,\rect^+_1]\times\ldots\times(-\infty,\rect^+_d]\times (\rect^+_1,\infty)\times\ldots\times (\rect^+_d,\infty)$ in $\Re^{4d}$ and the weight query threshold interval $I'=[\intervalL-\eps-\delta,\intervalU+\eps+\delta]$.
We repeat the following until no index is returned:
Using the range tree's $\rangetree$ query procedure $\rangetree.\mathsf{ReportFirst}(\rect',I')$ we compute, if any, a point $q_{(\rec,\hat{\rec})}$ in $Q\cap \rect'$ with weight in $I'$.
Let $j$ be the index of the synopsis $\Synop_{\points_j}$ that $q_{(\rec,\hat{\rec})}$ was constructed for. We report $j$, and we delete from $\rangetree$ all points $\{q_{(\rec,\bar{\rec})}\mid\rec\in \mathcal{R}_j\}$. We continue with the next iteration.
Let $J$ be the indexes we reported.
In the end of the process, we re-insert in $\rangetree$ all points we removed from $Q_j$ for every index $j\in J$.}

\paragraph{\textit{Example (cont.)}}
\new{
Let $\delta=0$, $\eps<0.01$, $\rect=[3,8]$ and $\interval=[0.2,0.4]$. We define the orthant $\rect'=[3,\infty)\times (-\infty,3)\times (-\infty, 8]\times (8,\infty)$ and the interval $I'=[0.2-\eps,0.4+\eps]$. The procedure $\rangetree.\mathsf{ReportFirst}(\rect', I')$ returns the point $(7,1,7,9)$ because it lies in $\rect'$ and its weight is $1/3\in[0.2-\eps,0.4+\eps]$. Hence, the index $1$ is reported.
On the other hand, the index $2$ is not reported. The only point from $Q_2$ that lies in $\rect'$ is the point $(4,2,6,10)$ that corresponds to the pair $([4,6], [2,10])\in \hat{\mathcal{R}}_2$. However, the weight of this point is $\frac{|[4,6]\cap \epsample_2|}{|\epsample_2|}=\frac{2}{4}=0.5>0.4+\eps$ so it is not returned. Indeed $\frac{|\rect\cap \points_2|}{|\points_2|}\geq \frac{|\rect\cap \epsample_2|}{|\epsample_2|}-\eps=0.5-\eps>0.4$, so index $2$ should not be reported.
}

\paragraph{Correctness}
We first start with a primitive lemma. It shows that whenever $\rangetree$ finds a point $q_{(\rec,\hat{\rec})}$ (with $\rec\in \mathcal{R}_j$), then $\rec$ is actually the largest rectangle in $\mathcal{R}_j$ with $\rec \subseteq \rect$. See Figure~\ref{fig:fig2}.
\begin{lemma}
\label{lem:techlem2}
 Let $j\in J$ be an index reported by the query procedure and let $q_{(\rec,\hat{\rec})}$ be the point that was found when $j$ was reported. There is no rectangle $\rec'\in \mathcal{R}_j$ such that $\rec\subset \rec'\subseteq \rect$.  
\end{lemma}
\begin{proof}
    Assume that there exists $\rec'\in \mathcal{R}_j$ such that $\rec\subset \rec'\subseteq \rect$. By the definition of the query orthant $\rect'$, we have $\rec\subseteq \rect\subset\!\!\subset \hat{\rec}$. Hence, we have $\rec\subset \rec'\subset\!\!\subset \hat{\rec}$. This is a contradiction because, by the definition of $Q_j$, a point $q_{(\rec,\hat{\rec})}$ exists if there is no $\rec''
    \in \mathcal{R}_j$ with $\rec\subset \rec''\subset\!\!\subset \hat{\rec}$.
\end{proof}

\begin{lemma}
\label{lem:HH2}
    If $\rec\in\mathcal{R}_j$ is the maximal rectangle among rectangles in $\mathcal{R}_j$, in the query rectangle $\rect$, then there exists $\hat{\rec}^{\rect}\in \mathcal{R}_j$ such that $q_{(\rec,\hat{\rec}^{\rect})}\in Q_j$, and $\rect\subset\!\!\subset \hat{\rec}^{\rect}$.
\end{lemma}
\begin{proof}
    Since there is no point in $Q$ that lies on an axis-align hyperplane that is tangent to $R$ we have that $\rec\subset \rect$. For every facet of $\rect$ we repeat: we expand the facet until we hit a point from $\epsample_j\cup\bar{\epsample}_j$. Notice that we will always hit such a point, either in $\epsample_j$ or on the boundary $\mathcal{B}$ (if we hit the boundary of $\mathcal{B}$, then by the construction of $\bar{\epsample}_j$ there will be at least a point from $\bar{\epsample}_j$ on the intersection of $\mathcal{B}$ and the extended facet).
    Let $\hat{\rec}^{\rect}$ be the new rectangle. By the construction of $\hat{\rec}^{\rect}$ and the fact that $\rec$ is maximal, we have that $\rect\subset\!\!\subset \hat{\rec}^{\rect}$, 
    $\hat{\rec}^{\rect}\in \mathcal{R}_j$, and there is no $\rec'\in\mathcal{R}_j$ such that $\rec\subset\rec'\subset\!\!\subset\hat{\rec}^{\rect}\in \mathcal{R}_j$.
\end{proof}

Let $\Pi=\pred_{\M_{\rect},\interval}$ be the logical expression with one range-predicate.
We show that we report all the ``correct'' indexes $\query_\Pi(\sets)$.

\begin{lemma}
\label{lem:tech1}
    $\query_\Pi(\sets)\subseteq J$, with probability at least $1-\prob$.
\end{lemma}
\begin{proof}
  From Lemma~\ref{lem:helper1}, we have that every $\epsample_i$ is an $(\eps+\delta)$-sample for $\points_i$ with probability at least $1-\prob$. Let $i\in \query_\Pi(\sets)$ be an index such that $\intervalL\leq \M_\rect(\points_i)\leq \intervalU$. We show that $i\in J$. Let $\rec$ be the largest/maximal rectangle in $\mathcal{R}_i$ such that $\rec\subseteq \rect$. 
  By definition, it holds that $\rec\cap \epsample_i = \rect\cap \epsample_i$. From Lemma~\ref{lem:HH2}, there exists a point $q_{(\rec,\hat{\rec}^{\rect})}\in Q_i\subseteq Q$.
  Hence, if we show that $q_{(\rec,\hat{\rec}^{\rect})}\in \rect'\cap Q$ and $\intervalL-\eps-\delta\leq w_{q_{(\rec,\hat{\rec}^{\rect})}}\leq \intervalU+\eps+\delta$, then the result follows. Indeed, by definition of $\epsample_i$, we have $\intervalL-\eps-\delta\leq \frac{|\points_i\cap \rect|}{|\points_i|}-\eps-\delta\leq \frac{|\epsample_i\cap \rect|}{|\epsample_i|}=\frac{|\epsample_i\cap \rec|}{|\epsample_i|}=\M_\rec(\epsample_i)=w_{q_{(\rec,\hat{\rec}^{\rect})}}$ and
  $\intervalU+\eps+\delta\geq \frac{|\points_i\cap \rect|}{|\points_i|}+\eps+\delta\geq \frac{|\epsample_i\cap \rect|}{|\epsample_i|}=\frac{|\epsample_i\cap \rec|}{|\epsample_i|}=\M_\rec(\epsample_i)=w_{q_{(\rec,\hat{\rec}^{\rect})}}$. Furthermore, from Lemma~\ref{lem:HH2},  $\rec\subseteq \rect$ and $\rect\subset\!\!\subset\hat{\rec}^{\rect}$ so
  we have $\rect^-_h\leq \rec^-_h$, $\rect^-_h> \hat{\rec}^{\rect-}_h$, $\rect^+_h\geq \rec^+_h$, and $\rect^+_h< \hat{\rec}^{\rect+}_h$,
  for every $h\in [d]$. Hence, by the definition of the orthant $\rect'$, it holds that $q_{(\rec,\hat{\rec}^{\rect})}\in \rect'\cap Q$.
\end{proof}

\begin{lemma}
\label{lem:tech2}
    For every $j\in J$ it holds that $ \intervalL-2\eps-2\delta\leq \M_\rect(\points_j)\leq \intervalU+2\eps+2\delta$, with probability at least $1-\prob$. 
\end{lemma}
\begin{proof}
    Let $j\in J$ be an index reported by our query procedure. Let $q_{(\rec,\hat{\rec})}\in Q_j$ be the point found in $\rect'$ at the moment we reported $j$. By definition, $w_{q_{(\rec,\hat{\rec})}}=\M_\rec(\epsample_j)=\frac{|\rec\cap \epsample_j|}{|\epsample_j|}\geq \intervalL-\eps-\delta$, and $w_{q_{(\rec,\hat{\rec})}}=\frac{|\rec\cap \epsample_j|}{|\epsample_j|}\leq \intervalU+\eps+\delta$, and from Lemma~\ref{lem:techlem2}, we know that $\rec$ is the largest rectangle inside $\rect$, so $|\rec\cap \epsample_j|= |\rect\cap \epsample_j|$. By definition of $\epsample_j$, we have $\M_\rect(\points_j)=\frac{|\points_j\cap \rect|}{|\points_j|}\geq \frac{|\rect\cap \epsample_j|}{|\epsample_j|}-\eps-\delta= \frac{|\rec\cap \epsample_j|}{|\epsample_j|}-\eps-\delta\geq \intervalL-2\eps-2\delta$, and $\M_\rect(\points_j)=\frac{|\points_j\cap \rect|}{|\points_j|}\leq \frac{|\rect\cap \epsample_j|}{|\epsample_j|}+\eps+\delta= \frac{|\rec\cap \epsample_j|}{|\epsample_j|}+\eps+\delta\leq \intervalU+2\eps+2\delta$.
\end{proof}

\begin{lemma}
\label{lem:duplG}
    The query procedure does not report duplicates.
\end{lemma}

\vspace{-0.8em}
\paragraph{Analysis}
\new{For simplicity, we assume that $\eps$ is a an arbitrarily small constant and $\prob=\frac{1}{\setsize}$. In Appendix~\ref{appndx:approxPercRangePred} we show the analysis for arbitrary values of $\eps$ and $\prob$.}
The proof of the next lemma can be found in Appendix~\ref{appndx:approxPercRangePred}.
\vspace{-0.5em}
\begin{lemma}
\label{lem:space2}
\new{
The data structure has $O(\setsize\cdot\log^{8d}(\setsize))$ space and it can be constructed in $O(\setsize\cdot\log^{8d}(\setsize))$ time.
The query procedure runs in $O(\log^{4d+1}(\setsize)+\out\cdot\log^{6d+1}(\setsize))$ time.}
\end{lemma}

\begin{theorem}
\label{thm:res2}
Let $\{\Synop_{\points_1},\ldots,\Synop_{\points_\setsize}\}$ be the input to the $\problemDI$ problem, such that $\Synop_{\points_i}$ is a synopsis of a dataset $\points_i\subset \Re^d$, where $d\geq 1$ is a constant, with $\err_{\Synop_{\points_i}}(\F^d_{\square})\leq \delta$ for $\delta\in[0,1)$, for every $i\in[\setsize]$. Let $\eps\in(0,1)$ be an arbitrarily small constant parameter.
A data structure of size $O(\setsize\cdot\log^{8d}(\setsize))$ can be constructed in $O(\setsize\cdot\log^{8d}(\setsize))$ time, such that
given a query $\Pi=\pred_{\M_\rect,\interval}$,
it returns a set of indexes $J$ such that $\query_\Pi(\sets)\subseteq J$ and for every $j\in J$, $\intervalL-\eps-2\delta\leq \M_\rect(\points_j)\leq \intervalU+\eps+2\delta$, with probability at least $1-\frac{1}{\setsize}$.
The query time is $O(\log^{4d+1}(\setsize)+\out\cdot\log^{6d+1}(\setsize))$.
\end{theorem}


\paragraph{Remark 1} The data structure from Theorem~\ref{thm:res2} can be made dynamic under insertion or deletion of synopses. Using the update procedure of the dynamic range tree, if a synopsis $\Synop_\points$ is inserted/deleted
we can update our data structure in $O(\log^{8d+1}(\setsize))$ time

\paragraph{Remark 2} If we do not know an upper bound $\delta$ on $\err_{\Synop_{\points_i}}(\F^d_{\square})$, then by slightly modifying the data structure above we can get the following result. Let $\err_{\Synop_{\points_i}}(\F^d_{\square})\leq \delta_i$ for unknown parameters $\delta_i\in[0,1)$.
Having the same complexities as in
Theorem~\ref{thm:res2}, we report a set of indexes $J$ such that, if $ \intervalL+\eps+\delta_i\leq \M_\rect(\points_i)\leq \intervalU-\eps-\delta_i$, then $i\in J$, and if $j\in J$, then $\intervalL-\eps-\delta_j\leq \M_\rect(\points_j)\leq \intervalU+\eps+\delta_j$.

\paragraph{Remark 3}
We note that the data structure we design satisfy delay guarantees as we report the indexes $J$, because range tree satisfies delay guarantees and we do not report duplicates. It is only in the end of the query procedure that we need to re-insert all the deleted points from $\rangetree$.
Using lazy updates~\cite{overmars1981worst, overmars1983design} similarly to~\cite{agarwal2021dynamic}, our data structure has bounded $O(\log^{6d+1}(\setsize))$ delay.

\vspace{-0.5em}
\section{Data structures for the \problemDIk Problem}
\label{sec:topK}
\vspace{-0.3em}
In this section, we design data structures for the 
\problemDIk problem. 
Recall that the \problemDIk problem is defined on threshold-predicates over an interval $\interval$ of the form $[\intervalL,1]$.
 Next, we show an approximate data structure for the \problemDI problem assuming only one predicate. In Appendix~\ref{subsec:topkpredm} we extend to an approximate data structure for the \problemDIk problem considering a logical expression over $\conj=O(1)$ predicates.
 Similarly to the previous section, we assume that for every $i\in[\setsize]$, $\err_{\Synop_{\points_i}}(\vec{\F}^d_{k}) \leq\delta_i\leq \delta$, for a known parameter $\delta\in [0,1)$. In the end, we give the guarantees assuming that $\delta_i$'s, are unknown.
 For simplicity, we assume that all points in $\bigcup_{i\in[\setsize]}\points_i$ lie in the unit ball.
 \new{Recall that each synopsis $\Synop_{\points}$ has a mechanism to return $\M_{\vector,k}(\Synop_\points)$, i.e., an estimation of the $k$-th largest score in $\points$ with respect to unit vector $\vector$.
 Let $\Synop_{\points}.\mathsf{Score}(\vector, k)$ be the synopsis' procedure that returns $\M_{\vector,k}(\Synop_\points)$.}
 

\paragraph{\new{Overview}}
\new{
The main idea follows from the data structure we designed for the \problemDI problem. If we knew the query vector $\vec{u}$, then the \problemDIk would be straightforward. 
Recall that for the \problemDI problem, we did not know the query rectangle $\rect$ upfront, so the first step was to create all possible combinatorially different rectangles generated from a set of samples. For the \problemDIk, we do not have samples, however, we construct a coreset using the notion of $\eps$-net. While we do not know the query vector $\vec{u}$, by constructing an $\eps$-net we get a set of vectors (of small size) such that for any query unit vector $\vec{u}$ there is always a unit vector in the $\eps$-net that lies very close to $\vec{u}$. While there are infinite query vectors, the intuition is that we can approximate the score of a point $p$ on a query vector $\vec{u}$ by computing the score of $p$ with respect to the vector of the $\eps$-net which is closest to $\vec{u}$. Hence, in the preprocessing phase, we use the synopsis of each dataset to estimate the score of the $k$-th highest score on every vector in the $\eps$-net. For each vector in the $\eps$-net, we store all the estimated costs (one from each dataset) into a geometric data structure. The computed estimated scores over all the vectors in the $\eps$-net consists of a coreset for our problem.
Given a query vector $\vec{u}$ and a one-sided interval $\interval=[\intervalL,1]$, we find the closest vector in the $\eps$-net and using the corresponding data structure we identify all datasets with scores in the range $[\intervalL-\eps-\delta,1]$. Because of the coreset construction, this technique returns a solution with theoretical guarantees for the \problemDIk problem.}

\begin{figure}
   \begin{minipage}{0.56\linewidth}
\begin{algorithm}[H]
  \caption{\new{$\mathsf{ConstructTop}(\Synop_{\points_1}, \ldots, \Synop_{\points_\setsize},\eps)$}}
  \label{alg:Topkconstruction}
   \small
 $\net\leftarrow$ $\eps$-net in $\Re^d$\;
 \ForEach{$\vector\in \net$}{
    $\Gamma_{\vector}=\emptyset$\;
    \ForEach{$i\in[\setsize]$}{
        $\gamma_{\vector}^{(i)}=\Synop_{\points_i}.\mathsf{Score}(\vector,k)$, \hspace{0.2em}
        $\Gamma_{\vector}=\Gamma_{\vector}\cup\{\gamma_{\vector}^{(i)}\}$\;
    }
    $\rangetree_{\vector}=\mathsf{RangeTreeConstruct}(\Gamma_{\vector})$\;
 }
 $\rangetree=\bigcup_{\vector\in\net}\rangetree_{\vector}$\;
 \Return \;
\end{algorithm}
\end{minipage}
\ \ \
\begin{minipage}{0.4\linewidth}
\vspace{-3.75em}
\begin{algorithm}[H]
  \caption{\new{$\mathsf{QueryTop}(\rangetree, \vec{u}, \interval, \eps,\delta)$}}
  \label{alg:Topkquery}
     \small
     $\vector=\argmin_{\vec{h}\in\net}||\vec{u}-\vec{h}||$\;
    $I'=[\intervalL-\eps-\delta,\infty)$\;
    $\Gamma=\rangetree_{\vector}.\mathsf{Report}(I')$\;
    \lForEach{$\gamma_{\vector}^{(j)}\in \Gamma$}{Report $j$}

  \Return\;
\end{algorithm}
\end{minipage}
\end{figure}


\paragraph{Data structure}
\new{In Algorithm~\ref{alg:Topkconstruction}, we show the pseudocode of the construction phase.}
We construct an $\eps$-net $\net$ on $\Re^d$.
For every $\vector \in\net$ and every $i\in[\setsize]$, we get the value $\gamma_{\vector}^{(i)}=\Synop_{\points_i}.\mathsf{Score}(\vector, k)$ as an approximation to $\score_k(\points_i,\vector)$.
Let $\Gamma_{\vector}=\{\gamma_{\vector}^{(i)}\mid i\in[\setsize]\}$. For every $\vector\in\net$, we construct a binary search tree ($1$-dimensional static range tree) $\tree_{\vector}$ on $\Gamma_{\vector}$.\footnote{Since $\gamma_{\vector}^{(i)}\in \Re$, we call it either a point in $\Re^1$ or a score/value.}


\paragraph{Query procedure}
\new{In Algorithm~\ref{alg:Topkquery}, we show the pseudocode of the query phase.}
Let $\vec{u}$ be the query vector and and $\interval$ be the one-sided query interval. We go through each vector in $\net$ and we compute $\vector=\argmin_{\vec{h}\in\net}||\vec{u}-\vec{h}||$. Using the tree $\tree_{\vector}$
we report all indexes $j\in[\setsize]$ such that $\gamma_{\vector}^{(j)}\geq \intervalL-\eps-\delta$. Let $J$ be the set of all indexes reported.

\paragraph{Correctness}
Let $\Pi=\pred_{\M_{\vec{u},k},\interval}$ be the logical expression with one threshold-predicate. Recall that $\query_\Pi(\sets)$ denotes the correct set of indexes we should report.

\begin{lemma}[\cite{agarwal2017efficient}]
\vspace{-0.5em}
\label{lem:helper00}
    For any point $p$ in the unit ball ($||p||\leq 1$) and any pair of unit vectors $\vector_1, \vector_2$ such that $||\vector_1-\vector_2||\leq \eps$, $|\score(p,\vector_1)-\score(p,\vector_2)|\leq \eps$.
\end{lemma}

\begin{lemma}
\label{lem:tech3}
   $\query_\Pi(\sets)\subseteq J$ and for every $j\in J$, $\score_k(\points_j,\vec{u})\geq \intervalL-2\eps-2\delta$.
\end{lemma}
\begin{proof}
Since $\Synop_{\points_i}$ is a synopsis for $\points_i$, by definition we have that $\score_k(\points_i,\vector)-\delta\leq \gamma^{(i)}_{\vector}\leq \score_k(\points_i,\vector)+\delta$ for every $i\in[\setsize]$.
First, we assume that an index $i\in \query_\Pi(\sets)$, so $\score_k(\points_i,\vec{u})\geq \intervalL$. We show that $i\in J$. By Lemma~\ref{lem:helper00}, we have $\score_k(\points_i,\vector)\geq \score_k(\points_i,\vec{u})-\eps\geq \intervalL-\eps$. By the definition of $\gamma^{(i)}_{\vector}$, we have $\gamma^{(i)}_{\vector}\geq \score_k(\points_i,\vector)-\delta\geq \intervalL-\eps-\delta$. Hence $i$ will be reported by the query procedure.
Next, we show that for every index $j\in J$, it holds $\score_k(\points_j,\vec{u})\geq \intervalL-3\eps-\delta$. The index $j$ is reported because $\gamma^{(j)}_{\vector}\geq \intervalL-\eps-\delta$. By the definition of $\gamma^{(i)}_{\vector}$, we have 
$\gamma^{(j)}_{\vector}\leq \score_k(\points_j,\vector)+\delta$. 
Furthermore, from Lemma~\ref{lem:helper00}, $\score_k(\points_j,\vec{u})\geq \score_k(\points_j,\vector)-\eps$.
Hence, 
$\score_k(\points_j,\vec{u})\geq \score_k(\points_j,\vector)-\eps\geq \gamma^{(i)}_{\vector}-\eps-\delta\geq \intervalL-2\eps-2\delta$.
\end{proof}
\vspace{-0.6em}
\paragraph{Analysis}
\new{For simplicity, we assume that $\eps$ is a an arbitrarily small constant. In Appendix~\ref{appndx:approxPercRangePred} we show the analysis for arbitrary values of $\eps$.}
Let $\timeSynop_{\Synop_{\points_i}}$ be the running time of $\Synop_{\points_i}.\mathsf{Score}(\vector, k)$.
The proof of the next lemma can be found in Appendix~\ref{appndx:topk}. We set $\eps\leftarrow \eps/2$ and we get the final result.
\vspace{-0.5em}
\begin{lemma}
\label{lem:space4}
\new{    The data structure has $O(\setsize)$ space and can be constructed in $O( \setsize\log(\setsize)+\sum_{i\in[\setsize]}\timeSynop_{\Synop_{\points_i}})$ time. The query time is $O(\log (\setsize) +\out)$.}
\end{lemma}

\vspace{-0.5em}
\begin{theorem}
\label{thm:resTopkAdd}
Let $\{\Synop_{\points_1},\ldots,\Synop_{\points_\setsize}\}$ be the input to the $\problemDIk$ problem, such that $\Synop_{\points_i}$ is a synopsis of a dataset $\points_i\subset \Re^d$, where $d\geq 1$ is a constant, with $\err_{\Synop_{\points_i}}(\vec{\F}^d_{k})\leq \delta$ for $\delta\in[0,1)$ and $k\geq 1$, for every $i\in[\setsize]$. Let $\eps\in(0,1)$ be an arbitrarily small constant parameter.
A data structure of size $O(\setsize)$ can be constructed in $O( \setsize\log(\setsize)+\sum_{i\in[\setsize]}\timeSynop_{\Synop_{\points_i}})$ time, such that
given a query $\Pi=\pred_{\M_{\vec{u},k},\interval}$ for a one-sided interval $\interval=[\intervalL,1]$,
it returns a set of indexes $J$ such that $\query_\Pi(\sets)\subseteq J$ and for every $j\in J$, $\M_{\vec{u},k}(\points_j)\geq \intervalL-\eps-2\delta$.
The query time is $O(\log (\setsize) +\out)$.

\end{theorem}

\vspace{-1em}
\section{Future work}
\label{sec:concl}

\new{
There are multiple interesting open problems arising from our data structures and the proposed framework in general.


\newcommand{\dist}{\mathsf{dist}}
$\bullet$ Our new framework can be used to define multiple types of distribution-aware queries such as nearest neighbor or diversity queries
both in the centralized and federated setting. We are given a repository $\points_1,...,\points_N$.
For nearest neighbor queries: given a query point $q$ and a threshold $\tau$, return all datasets $\points_j$ such that $\dist(q,\points_j)\leq\tau$, where $\dist(q,\points_j)$ is the closest distance between q and a point in $\points_j$. For diversity queries: given a query rectangle $\rect$ and a threshold $\tau$, return all datasets $\points_j$ such that $\mathsf{div}(\points_j\cap \rect)\geq \tau$, where $\mathsf{div}()$ is a diversity function~\cite{indyk2014composable}.
While the goal of the paper is not to provide a data structure with theoretical guarantees for every problem arising from our proposed framework, our techniques are useful to handle additional queries such as nearest neighbor and diversity queries. 
The main obstacle for these queries is that, to the best of our knowledge, there are no high-quality and small-size coresets we can use for the construction of our data structures. 
Recall that for the \problemDI problem we used the notion of $\eps$-sample, while for the \problemDIk problem we used the notion of $\eps$-net to construct coresets that are sufficient to answer any percentile or preference query, respectively.
On the other hand, for the approximate nearest neighbor problem (similar arguments can be made for other types of queries), given a set of points $\points$, there is no method to derive a subset $S$, where $|S|\ll|\points|$, such that for any query point $q\in \Re^d$, $\dist(q,S)\leq(1+\eps)\cdot\dist(q,\points)$. There are some recent results constructing coresets for approximate nearest neighbor queries with additive approximation~\cite{gao2024rabitq}. 
While our current data structures only handle percentile and preference queries (percentile measure and top-k preference measure functions), our work is a first step toward efficiently solving general problems arising from our proposed framework.

$\bullet$ For the \problemCDI problem in $\Re$ we proposed an exact data structure of near linear space and polylogarithmic query, assuming that $\interval$ is fixed. A natural question to ask is whether this data structure can be extended to allow $\interval$ to be part of the query.
Our current data structure (in Appendix~\ref{appndx:lb1}) heavily uses the fact that $\interval=[\intervalL, \intervalU]$ is known in the preprocessing phase.
One idea to overcome this obstacle is to construct the data structure for different dyadic ranges in $[0,1]$. Such an idea has been used in~\cite{cormode2011synopses} for range queries. Given $\rect$ and $\interval$, we could combine the different results from the precomputed dyadic ranges to have a valid solution for our problem. Currently, it is not clear how the precomputed dyadic ranges can be combined to get an exact solution.
}

\bibliographystyle{abbrv}
\bibliography{ref}
\newpage
\appendix

\section{Missing details from Section~\ref{sec:prelim}}
\label{appndx:prelim}
\begin{proof}[Proof of Lemma~\ref{lem:helper1}]
   \new{ Since $\err_{\Synop_X}(\F_{\square}^d)=\delta$, we have that for every rectangle $\rect\in\Re^d$, $|\M_{\rect}(X) - \M_{\rect}(\Synop_{X})|\leq \delta\Leftrightarrow -\delta+\M_{\rect}(\Synop_{X})\leq \M_{\rect}(X)\leq \delta+\M_{\rect}(\Synop_{X})\Leftrightarrow -\delta+\frac{|\rect\cap \Synop_{X}|}{|\Synop_X|}\leq \frac{|\rect\cap X|}{|X|}\leq \delta +\frac{|\rect\cap \Synop_{X}|}{|\Synop_X|}$. Since $\epsample$ is a set of $O(\eps^{-2}\log \prob^{-1})$ random samples from $\Synop_X$, from~\cite{chazelle2000discrepancy}, we have that $\epsample$ is an $\eps$-sample for $\Synop_X$ with probability at least $1-\prob^{-1}$. Hence, by definition, $-\eps+\frac{|\rect\cap \epsample|}{|\epsample|}\leq \frac{|\rect\cap \Synop_X|}{|\Synop_X|}\leq \eps+\frac{|\rect\cap \epsample|}{|\epsample|}$. By combining the derived inequalities, $-\delta-\eps+\frac{|\rect\cap \epsample|}{|\epsample|}\leq \frac{|\rect\cap X|}{|X|}\leq \delta+\eps+\frac{|\rect\cap \epsample|}{|\epsample|}\Leftrightarrow \left|\frac{|\rect\cap X|}{|X|}-\frac{|\rect\cap \epsample|}{|\epsample|}\right|\leq \delta+\eps$. The result follows.
   }
\end{proof}

\section{Missing details from Section~\ref{sec:lb}}
\label{appndx:lb}
\begin{proof}[Proof of Lemma~\ref{lem:rest}]
Throughout this proof, we use the term \textit{total size} of a collection of sets to denote the sum of the sizes of the sets in the collection. 

    We show that if a data structure $\D$ exists for the uniform set intersection problem with space near-linear on total size and query time near-linear on output size, we can break the strong set intersection conjecture. Intuitively, we show that if $\D$ exists, we can use it as a black box to efficiently answer general set intersection queries. 

    Let $S = \{S_1, \ldots, S_g \}$, be an arbitrary collection of sets without any restrictions over a universe $\U$, with total size $W = \sum_{S_i \in S}|S_i|$. We build a data structure $\D'$ to answer intersection queries over $S$ as follows. For any element $u \in \U$, let $c_u = |S_i \in S : u \in S_i|$ be the number of different sets that $u$ belongs to. First, we partition $\U$ into three subsets, the \textit{light} elements $\U^l = \{u \in \U : c_u < W^{\frac{5}{12}}\}$, the \textit{medium} elements $\U^m = \{u \in \U : W^{\frac{5}{12}} \leq c_u < W^{\frac{5}{6}}\}$, and the \textit{heavy} elements $\U^h = \{u \in \U : W^{\frac{5}{6}} \leq c_u\}$.
    Next, for every pair $i,j\in[g]$, such that $S_i \cap S_j \cap \U^l \neq \emptyset$ we store the list $L_{i,j}=S_i \cap S_j \cap \U^l$. We place all computed lists in a simple lookup table with respect to the indexes $i, j$. For each $S_i \in S$, we define $T_i = S_i \cap \U^m$. Let $T$ be an initially empty collection of sets. We add $T_i$ to $T$ for every $S_i \in S$. Our goal is to make $T$ uniform and to do so, we add dummy sets to $T$. More formally, for any element $u \in \U^m$, let $c'_u = |\{T_i \in T : u \in T_i\}|$ be the number of different sets in $T$ that $u$ belongs to. We add $W^{\frac{5}{6}} - c'_u$ identical sets $\{u\}$ to $T$ for all $u \in \U^m$ to make $T$ uniform. Let $M$ be the total size of $T$ after these operations. We build the data structure $\D$ over $T$. Since $T$ is uniform, by the assumption, $\D$ has size $\O(M)$ and answers set intersection queries over $T$ in $\O(1 + \outt)$ time. At this stage, we are done with the preprocessing of $\D'$ and we are ready to answer set intersection queries over $S$. 

    When an intersection query for a pair of sets $S_i, S_j$ comes, we first look for the pair $i, j$ in the lookup table to see whether $S_i$ and $S_j$ intersect on a light value. If the answer is positive we report the elements in $L_{i,j}$. Next, 
    we run the intersection query for $T_i$ and $T_J$ on $\D$
     and report all the medium elements in $S_i \cap S_j$ and continue. Finally, we go through each heavy element $u \in \U^h$ one by one and check whether both $S_i$ and $S_j$ contain $u$ and report it if the answer is positive. 
    
    It is easy to verify that by the described procedure, we correctly answer the queries since we report the light values in $S_i\cap S_j$ using the preprocessed lookup table, the medium values in $S_i\cap S_j$ invoking $\D$, and the heavy values in $S_i\cap S_j$ considering every heavy value one by one.

    Now, we to analyze the space and query time of $\D'$ as described above.
    For the query phase, we check the lookup table in $\O(1)$ time and report the light elements in $S_i\cap S_j$ in $\O(\outt)$ time. We invoke $\D$ once and spend $\O(\outt)$ time to report the heavy elements in $S_i\cap S_j$. We check all heavy elements in $\O(|\U^h|)$ time. We have $\sum_{u \in \U}c_u = \sum _{S_i \in S}|S_i| = W$, and for each heavy value $u \in \U^h$, we know $c_u \geq W^{\frac{5}{6}}$ by definition. Thus, we have $|\U^h| \leq \frac{W}{W^{\frac{5}{6}}} = W^{\frac{1}{6}}$. Therefore, $\D'$ answers the intersection queries in $\O(W^{\frac{1}{6}} + \outt)$ time.

    Next, we bound the space of $\D'$.
    For each medium element $u \in \U^m$, we have $c_u \geq W^{\frac{5}{12}}$ by definition and hence we have $|\U^m| \leq \frac{W}{W^{\frac{5}{12}}} = W^{\frac{7}{12}}$. By the construction of $T$, we know that each $u \in \U^m$ belongs to exactly $W^{\frac{5}{6}}$ different sets in $T$, so the total size of $T$ is equal to $|\U^m|W^{\frac{5}{6}} \leq W^{\frac{7}{12}}W^{\frac{5}{6}} = W^{\frac{17}{12}}$. So, by the assumption, $\D$ uses $\O(W^\frac{17}{12})$ space.
    To bound the size of the lookup table, we need to bound the number of the pairs of sets that intersect on a light element. We have $\sum_{S_i \in S}|S_i| = W$ and hence there are at most $W$ pairs of set-elements like $(S_i, u)$ such that $S_i \in S$, $u \in S_i$, and $u \in \U^l$. By definition, we know that for each light element $u$, it holds that $c_u < W^{\frac{5}{12}}$ and hence each of these set-element pairs can be matched to at most $W^{\frac{5}{12}}$ other set-element pair on the same element to create a pair of sets that intersect on a light value. Thus, there are at most $W^{\frac{5}{12}}W = W^{\frac{17}{12}}$ different set-element, set-element pairs where the element is a single light value, so $\sum_{\ell_1,\ell_2\in[g]}|L_{\ell_1,\ell_2}|=\O(W^\frac{17}{12})$ and the size of 
    the lookup table is $\O(W^\frac{17}{12})$.
    
    Overall, $\D'$ uses $\O(W + W^\frac{17}{12} + W^\frac{17}{12}) = \O(W^\frac{17}{12})$ space to answer queries in $\O(W^\frac{1}{6} + \outt)$ time over an input $S$ with total size of $W$.
    According to Conjecture~\ref{conj:inter}, any data structure that answers intersection queries in $\O(W^{\frac{1}{6}}+ \outt)$, needs space of $\tilde{\Omega}\left(\frac{W^2}{W^{\frac{2}{6}}}\right) = \tilde{\Omega}(W^{\frac{5}{3}})$. This implies that $\D$ does not exist since $W^{\frac{5}{3}} > W^{\frac{17}{12}}$, and the result follows.
\end{proof}

\subsection{Reduction from Set Intersection to \problemCDI}
\label{appndx:lb1}

We show that having an efficient data structure for answering the \problemCDI query contradicts Lemma~\ref{lem:rest} and hence the existence of efficient data structures for the \problemCDI problem is unlikely. 

Let $S_1,\ldots, S_g$ be a uniform collection of sets. Let $M = \sum_{i \in [g]}|S_i|$ and let $s_{i,k}$ be the value of the $k$-th item in $S_i$ (we consider any arbitrary order of the items in each $S_i$). Without loss of generality, we assume that $\U = \cup_{i \in [g]} S_i = \{1, \ldots, q\}$, for some integer $q \leq M$. 

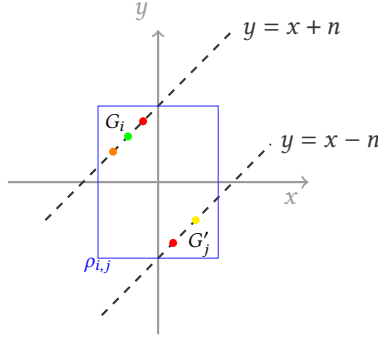
\begin{figure}[t]
\centering
\begin{tikzpicture}[scale=1]

  \draw[->, thick, gray!80] (-2, 0) -- (2, 0) node[below left] {$x$};
  \draw[->, thick, gray!80] (0, -2) -- (0, 2) node[above left] {$y$};

  \draw[black!80, thick, dashed] (-1.5, -0.5) -- (1, 2) node[anchor=west] {$y = x + n$};
  \draw[black!80, thick, dashed] (-0.75, -1.75) -- (1.5, 0.5) node[anchor=west] {$y = x - n$};

  \draw[blue!80] (-0.8, 1.0) rectangle (0.8, -1.0);

  \fill[red] (-0.2, 0.8) circle (1.5pt) node[anchor=east] {};
  \fill[green] (-0.4, 0.6) circle (1.5pt);
  \fill[orange] (-0.6, 0.4) circle (1.5pt);
  \node[anchor=north west, scale = 0.8, black] at (-0.8, 1.00) {$G_i$};

  \fill[yellow] (0.5, -0.5) circle (1.5pt);
  \fill[red] (0.2, -0.8) circle (1.5pt);
  \node[anchor=south east, scale = 0.8] at (0.8, -1.08) {$G'_j$};

  \node[anchor=south, scale = 0.8, blue] at (-0.8, -1.3) {$\rho_{i,j}$};

\end{tikzpicture}
\caption{An example of the lower bound construction. The rectangle $\rho_{i, j}$ only contains the points from $G_i$ and $G'_j$. The points from same the same point sets are shown with the same color. Since red appears twice in $\rho_{i, j}$, the sets $S_i$ and $S_j$ intersect.}
    \label{fig:lower}
\end{figure}

Using the idea from \cite{rahul2012algorithms}, we design an instance $\sets=\{\points_1,\ldots, \points_{q}\}$ of the \problemCDI problem in $\Re^2$ based on these sets. All points lie on two parallel lines $L=x+M$, and $L'=x-M$. Let $m_0=0$ and $m_i=m_{i-1}+|S_i|$ for $i=1,\ldots, g$. For each $s_{i,k}$ we create two points,
$p_{i,k}=(-(k+m_{i-1}), -(k+m_{i-1})+M)$ on $L$, and $p_{i,k}'=((k+m_{i-1}), k+m_{i-1}-M)$ on $L'$ and we add both points $p_{i,k}, p_{i,k}'$ to the point set $\points_{s_{i,k}}$. Let $G_i$ be the set of points corresponding to $S_i$ that lie on $L$, and $G_i'$ be the set of points corresponding to $S_i$ that lie on $L'$. Let $H = \cup_{i \in [q]}\points_i$ denote the set of all the points we created.

Given any pair of indices $i,j \in [g]$, in the set intersection problem, we want to report $S_i \cap S_j$. Geometrically, it is known~\cite{rahul2012algorithms} that we can find a rectangle $\rho_{i,j}$ in $O(1)$ time such that $\rho_{i,j}\cap H = G_i \cup G'_j$ (see Figure~\ref{fig:lower}). Hence, it is easy to check that for any element $u \in \U$, we have $u \in S_i \cap S_j$, if and only if $|\points_u \cap \rho_{i,j}| = 2$.
Moreover, since we started from a uniform collection of sets, all the point sets contain the same number of points, i.e., for any pair of point sets $\points_i, \points_j \in \sets$, we have $|\points_i| = |\points_j| = t$, for some number $t$. To answer the intersection queries, we build a data structure $\D$ for answering \problemCDI query over $\sets$. For an arbitrary pair of sets $S_i, S_j$, we first find $\rho_{i,j}$ and then query $\D$ with rectangle $\rho_{i, j}$ and interval $[\frac{1.5}{t}, 1]$. Note that $\D$ reports an index $r$, if and only if $|\points_r \cap \rho_{i,j}| > 1$ (so $|\points_r \cap \rho_{i,j}| = 2$), and hence, $\D$ reports $S_i \cap S_j$. 

For an input of size $x$, let $\mathcal{S}(x)$ be the space used by a data structure for the \problemCDI and $\mathcal{Q}(x)+\outt$ be its query time, where $\outt$ is the output size.
By the above construction, we can answer uniform set intersection queries in time $O(1)+\mathcal{Q}(M)+\outt$ and space $\mathcal{S}(M) + M$, using $\D$. 
Along with Lemma~\ref{lem:rest}, we conclude with Theorem~\ref{lowerbound:centralized-percentile}.

\subsection{Reduction from Halfspace reporting to \problemCDIk}
\label{appndx:lb2}
We show a reduction from the halfspace reporting problem to \problemCDIk problem.
Let $U$ be a set of $n$ points in $\Re^d$.
Without loss of generality, we can assume that they all lie in the first orthant (otherwise we can apply a transformation to move all the points to the first orthant without changing the definition of the problem). Furthermore, without loss of generality, we can assume that $U$ lies in the unit ball.
If not, then we can apply a scaling transformation to map $U$ to the set $\bar{U}$ in the unit ball. It is known that for any halfspace $H$ we can construct a halfspace $\bar{H}$ in $O(1)$ time, such that $u\in U\cap H$ if and only if $\bar{u}\in \bar{U}\cap \bar{H}$.
Let $\sets=\{\points_1,\ldots, \points_n\}$ be a repository such that for every $i\in[n]$, the dataset $\points_i$ contains only the point $u_i\in U$.
Let $\D$ be a data structure for the \problemCDIk problem over $\sets$. We also apply the following affine transformation to $U$. Let $x\in \Re^d$ be a point such that all coordinates of the points in $U$ are smaller than $x$'s coordinates. We apply a rotation transformation to $U$ creating $U'$ such that $x$ becomes the origin and $U'$ lies in the first orthant. Such a rotation is computed in $O(n)$ time.
Let $\sets'=\{\points_1',\ldots, \points_n'\}$ be a repository such that for every $i\in[n]$, the dataset $\points_i'$ contains only the point $u_i'\in U$.
Let $\D'$ be a data structure for the \problemCDIk problem over $\sets'$.

Now assume that a query halfspace $H$ is given for the halfspace reporting problem. 
We check in $O(1)$ time whether $H$ intersects the first orthant. If not, we return $\emptyset$. Otherwise, we proceed as follows.
We check in $O(1)$ time whether the origin belongs in $H$. If not, then: Let $h$ be the halfplane that defines $H$. We compute in $O(1)$ time a unit vector $\vector$ perpendicular to $h$. We also compute in $O(1)$ time the line $L$ that supports $\vector$ and passes through the origin. Let $\intervalL$ be the distance from the origin to the intersection point of $L$ and $h$.
Let $\interval=[\intervalL,1]$ and $k=1$.
We query $\D$ on predicate $\pred_{\M_{\vector,k},\interval}$.
On the other hand, if the origin belongs in $H$, then we apply the same rotation we had from $U$ to $U'$, to the halfspace $H$ to get $H'$, in $O(1)$ time.
We repeat the same procedure as before to compute $\vector'$ and $\intervalL'$. Let $\interval'=[\intervalL',1]$ and $k'=1$.
We query $\D'$ on predicate $\pred_{\M_{\vector',k'},\interval'}$.

When $H$ does not contain the origin, it is easy to see that $u_i\in H$ if and only if $\pred_{\M_{\vector,k},\interval}(\points_i)=\textsf{True}$. Similarly, if $H$ contains the origin, $u_i\in H$ if and only if $u_i'\in H'$. And $u_i'\in H'$ if and only if $\pred_{\M_{\vector',k'},\interval'}(\points_i')=\textsf{True}$.

Let $\mathcal{S}(n)$ be the space of the data structure for the \problemCDIk problem on $n$ datasets and $\mathcal{Q}(n)+\outt$ be its query time, where $\outt$ is the output size. Using $\D, \D'$, we can solve the halfspace reporting problem with a data structure of $2\mathcal{S}(n)$ space and $O(1)+\mathcal{Q}(n)+\outt$ query time.

\section{Missing details from Section~\ref{sec:Perc}}
\label{appndx:secPerc}

\subsection{Exact data structure for \problemCDI in $\Re^1$}
\label{subsec:percExact}
In the \problemCDI problem, we have full access to $\points_i$ for every $i\in[\setsize]$.
We assume that we know the range-predicate $\interval$ upfront, i.e., $\interval=[\intervalL,\intervalU]$ is not part of the query. This is only a requirement in this subsection. In all the other general approximate data structure $\interval$ is part of the query.
The data structure works for points in $\Re^1$.

\paragraph{Main idea}
The main issue to answer a \problemCDI query is that many points from the same dataset might lie inside the query rectangle (interval in $\Re^1$) $\rect$. Instead, for every dataset $\points_i$ we map all its points to a new dataset in $\Re^4$, such that given $\rect$, there exists an orthant $\rect'$ in $\Re^4$ satisfying the following strong property: For every dataset there will be at most one point in $\rect'$, and for every dataset $\points_i$ that has a point in $\rect'$ we guarantee that $i$ must be reported to answer the \problemCDI query correctly.



\paragraph{Data structure}
Let $\points_i=\{p_1^{(i)}, \ldots, p_{\pointsize_i}^{(i)}\}$ for every $\points_i\in \sets$. Without loss of generality, we also assume that $p^{(i)}_{h_1}\neq p^{(i)}_{h_2}$ for every $i\in [\setsize]$ and $h_1\neq h_2\in [\pointsize_i]$. For every $i\in [\setsize]$, we sort the points in $\points_i$ in ascending order. Without loss of generality assume $p_1^{(i)}< \ldots < p_{\pointsize_i}^{(i)}$.
For every point $p^{(i)}_{j}\in \points_i$, let $q_j$ be the point in $\points_i$ such that $|[q_j,p^{(i)}_j]\cap \points_i|=\intervalU\cdot \pointsize_i$, let $r_j$ be the point in $\points_i$ such that $|[r_j,p^{(i)}_j]\cap \points_i|=\intervalL\cdot \pointsize_i$, and let $s_j=p^{(i)}_{j+1}$.
We straightforwardly handle the corner cases, for example if $s_j$ is not defined because $p^{(i)}_{j+1}$ is the rightmost point in $\points_i$, we can assume that $s_j\rightarrow\infty$.
Let $Q_i=\{(q_j,r_j,p^{(i)}_j,s_j)\mid p^{(i)}_j\}$ be the set of $\pointsize_i$ four dimensional points.
Let $Q=\bigcup_{i\in [\setsize]}Q_i$.
We construct a (static) range tree $\mathcal{T}$ for range reporting queries on $Q$.

\paragraph{Query procedure}
Let $\rect=[\rect^-,\rect^+]$ be the query interval in $\Re^1$.
We define the orthant $\rect'=(-\infty,\rect^-]\times [\rect^-,\infty)\times (-\infty,\rect^+]\times [\rect^+,\infty)$ in $\Re^4$ and we run a reporting query on $\rangetree$. Each time the query procedure finds a point that belongs in $Q_i$, we report the index $i$.

\paragraph{Correctness}
\begin{lemma}
\label{lem:dupl0}
    The query procedure does not report duplicates
\end{lemma}
\begin{proof}
    For every $i\in [\setsize]$, we show that there is at most one point $p_j^{(i)}$ such that $(q_j,r_j,p^{(i)}_j,s_j)\in \rect'\cap Q_i$. We observe that only the point in $\points_i$ with the largest coordinate in $\rect$ satisfies the last two linear constraints of the orthant. If $p_j^{(i)}$ is the point with the largest coordinate in $\rect$ among the points in $\points_i$ then $\points^{(i)}_j\leq \rect^+$ and $s_j>\rect^+$. For any $p^{(i)}_h$ with $h<j$ it does not hold that $s_j\geq \rect^+$, while for any $p^{(i)}_h$ with $h>j$ it does not hold that $p^{(i)}_h\leq \rect^+$.
\end{proof}

Let $\Pi=\pred_{\M_{\rect},\interval}$ be the logical expression with the range-predicate $\interval$. Recall that $\query_\Pi(\sets)$ denotes the correct set of indexes we should report.


\begin{lemma}
    The query procedure returns $\query_\Pi(\sets)$. 
\end{lemma}
\begin{proof}
    Let $J$ be the set of indexes reported by the query procedure. Let $i\in \query_\Pi(\sets)$. Let $p^{(i)}_j$ be the points in $\points_i\cap \rect$ with the largest coordinate. From Lemma~\ref{lem:dupl0}, the last two linear constraints of $\rect'$ are satisfied. Since $i\in\query_\Pi(\sets)$ it must be the case that $\intervalL\leq \frac{|\points_i\cap \rect|}{|\points_i|}\leq \intervalU\Leftrightarrow \intervalL\cdot\pointsize_i\leq |\points_i\cap \rect|\leq \intervalU\cdot\pointsize_i\Leftrightarrow q_j\leq \rect^- \text{ and } r_j\geq \rect^-$. Hence, $(q_j,r_j, p^{(i)}_j,s_j)\in \rect'$ and the index $i$ will be reported by the query procedure.
    Equivalently, if $i\notin \query_\Pi(\sets)$, then $(q_j,r_j, p^{(i)}_j,s_j)\notin \rect'$ and the index $i$ will not be reported by the query procedure. 
\end{proof}

\paragraph{Analysis}
\begin{lemma}
    The data structure has $O(\totalsize\log^3( \totalsize))$ space and it can be constructed in $O(\totalsize\log^3(\totalsize))$ time.
\end{lemma}
\begin{proof}
    For every point $p_j^{(i)}$ in $\points_i$ we add one point in $Q_i$. The result follows from the complexity of constructing and storing a static range tree on $\totalsize$ $4$-dimensional points.
\end{proof}

\begin{lemma}
    The query time is $O(\log^3(\setsize\cdot\pointsize)+\out)$.
\end{lemma}
\begin{proof}
    It follows form the proof of Lemma~\ref{lem:dupl0} and from the query procedure of a $4$-dimensional range tree.
\end{proof}

\begin{theorem}
\label{thm:res0}
Let $\sets=\{\points_1,\ldots,\points_\setsize\}$ be the input to the $\problemCDI$ problem, such that $\totalsize=\sum_{i\in[\setsize]}|\points_i|$, and $\points_i\subset \Re$ for every $i\in[\setsize]$. 
Let $\interval$ be an interval in $\Re$.
A data structure of size $O(\totalsize\log^3(\totalsize))$ can be constructed in $O(\totalsize\log^3(\totalsize))$ time, such that
given a query $\Pi=\pred_{\M_\rect,\interval}$, it returns the set of indexes $\query_\Pi(\sets)$ in $O(\log^3(\totalsize)+|\query_\Pi(\sets)|)$ query time. 

\end{theorem}

\paragraph{Remark 1} The data structure in Theorem~\ref{thm:res0} can be made dynamic under insertion or deletion of datasets or points in the datasets. Instead of using a static range tree, we would use a dynamic range tree as discussed in Section~\ref{sec:prelim}. If a new dataset $\points$ with $n$ points is inserted (resp. deleted) in (resp. from) $\sets$, then the data structure is updated in $O(n\log^4(\totalsize+n)$ time. Furthermore, if a point is inserted/deleted from a current dataset, then the data structure can be updated in $O(\log^4\totalsize)$ time.

\paragraph{Remark 2} The range tree supports delay guarantees. Since we do not report duplicates, we can report $\query_\Pi(\sets)$ with delay guarantee $O(\log^3\totalsize)$.

\subsection{Missing proofs from Subsection~\ref{subsec:approxpercopenright}}
\label{appndx:approxpercopenright}

\begin{proof}[Proof of Lemma~\ref{lem:space1}]
For simplicity, we assume that we can get one sample from $\Synop_{\points_i}$ in $O(1)$ time, for every $i\in[\setsize]$. 
We have $|\epsample_i|=O(\eps^{-2}\log(\setsize\prob^{-1}))$, so there are $O(\eps^{-4d}\log^{2d}(\setsize\prob^{-1}))$ combinatorially different hyper-rectangles defined by $\epsample$, i.e., $|\mathcal{R}_i|=O(\eps^{-4d}\log^{2d}(\setsize\prob^{-1}))$. Hence, $|Q|=O(\setsize\cdot\eps^{-4d}\log^{2d}(\setsize\prob^{-1}))$ and it is constructed in $O(\setsize\cdot\eps^{-4d}\log^{2d}(\setsize\prob^{-1}))$ time. Using an additional static range tree on $\epsample_i$ for counting queries, we also compute all the values in $W$ in $O(\setsize\cdot\eps^{-4d}\log^{2d-1}(\setsize\prob^{-1})\log^{d}(\eps^{-1}\log(\setsize\prob^{-1})))$ time. After computing $Q$ and $W$, the range tree $\mathcal{T}$ is constructed in $O(\setsize\cdot\eps^{-4d}\log^{2d}(\setsize\prob^{-1}) \log^{2d}(\setsize\cdot\eps^{-1}\log(\setsize\prob^{-1})))$ time and it has 
$O(\setsize\cdot\eps^{-4d}\log^{2d}(\setsize\prob^{-1}) \log^{2d}(\setsize\cdot\eps^{-1}\log(\setsize\prob^{-1})))$ space. Overall, our data structure is constructed in $O(\setsize\cdot\eps^{-4d}\log^{2d}(\setsize\prob^{-1}) \log^{2d}(\setsize\cdot\eps^{-1}\log(\setsize\prob^{-1})))$ time and it has $O(\setsize\cdot\eps^{-4d}\log^{2d}(\setsize\prob^{-1}) \log^{2d}(\setsize\cdot\eps^{-1}\log(\setsize\prob^{-1})))$ space.

Next, we focus on the query time.
Every query in $\rangetree$ takes $O(\log^{2d+1}(\setsize\cdot\eps^{-1}\log(\setsize\prob^{-1})))$ time. Every update operation takes $O(\log^{2d+1}(\setsize\cdot\eps^{-1}\log(\setsize\prob^{-1})))$ time. Each time we report an index $j$, we delete all points in $Q_j$. From Lemma~\ref{lem:space1}, we have $|Q_j|=O(\eps^{-4d}\log^{2d}(\setsize\prob^{-1}))$. The overall query time is $O(\log^{2d+1}(\setsize\cdot\eps^{-1}\log(\setsize\prob^{-1}))+\out\cdot \eps^{-4d}\log^{2d}(\setsize\prob^{-1})\log^{2d+1}(\setsize\cdot\eps^{-1}\log(\setsize\prob^{-1})))$.
\end{proof}

\subsection{Missing proofs from Subsection~\ref{subsec:approxPercRangePred}}
\label{appndx:approxPercRangePred}


\begin{proof}[Proof of Lemma~\ref{lem:duplG}]
Notice that every time we report an index $j$, we do not remove all points in $Q_j$. Instead, we only remove the points $\{q_{(\rec,\bar{\rec})}\mid\rec\in \mathcal{R}_j\}$. From Lemma~\ref{lem:techlem2}, any time that our query procedure finds $q_{(\rec,\hat{\rec})}$, $\rec$ must be the largest hyper-rectangle inside $\rect$. Hence, by removing all points $\{q_{(\rec,\bar{\rec})}\mid\rec\in \mathcal{R}_j\}$, the index $j$ will not be reported again.
\end{proof}

\begin{proof}[Proof of Lemma~\ref{lem:space2}]
We assume that we get one sample from $\Synop_{\points_i}$ in  $O(1)$ time, for every $i\in[\setsize]$. 
We have $|\epsample_i|\leq |\bar{\epsample}_i|=O(\eps^{-2}\log(\setsize\prob^{-1}))$.
So
$$|\mathcal{R}_i|=O(\eps^{-4d}\log^{2d}(\setsize\prob^{-1})), \quad\text{and }\quad |\hat{\mathcal{R}}_i|=O(\eps^{-8d}\log^{4d}(\setsize\prob^{-1})).$$
Hence, $|Q|=O(\setsize\cdot\eps^{-8d}\log^{4d}(\setsize\prob^{-1}))$ and it is constructed in $$O(\setsize\cdot\eps^{-8d}\log^{4d}(\setsize\prob^{-1})\log^{2d}(\eps^{-1}\log(\setsize\prob^{-1})))$$ time, using a range tree to check whether there exists $\rec'$ such that $\rec\subseteq \rec'\subseteq \hat{\rec}$.
Using an additional static range tree on $\epsample_i$ for counting queries, we also compute $W$ in $$O(\setsize\cdot\eps^{-8d}\log^{4d}(\setsize\prob^{-1})\log^d(\eps^{-1}\log(\setsize\prob^{-1})))$$ time.
After computing $Q$ and $W$, the range tree $\mathcal{T}$ is constructed in $O(\setsize\cdot\eps^{-8d}\log^{4d}(\setsize\prob^{-1}) \log^{4d}(\setsize\cdot\eps^{-1}\log(\setsize\prob^{-1})))$ time and it has 
$O(\setsize\cdot\eps^{-8d}\log^{4d}(\setsize\prob^{-1}) \log^{4d}(\setsize\cdot\eps^{-1}\log(\setsize\prob^{-1})))$ space. Overall, our data structure is constructed in $O(\setsize\cdot\eps^{-8d}\log^{4d}(\setsize\prob^{-1}) \log^{4d}(\setsize\cdot\eps^{-1}\log(\setsize\prob^{-1})))$ time and it has $O(\setsize\cdot\eps^{-8d}\log^{4d}(\setsize\prob^{-1}) \log^{4d}(\setsize\cdot\eps^{-1}\log(\setsize\prob^{-1})))$ space.

Next, we focus on the query time.
Every query in $\rangetree$ takes $O(\log^{4d+1}(\setsize\cdot\eps^{-1}\log(\setsize\prob^{-1})))$ time. Every update operation takes $O(\log^{4d+1}(\setsize\cdot\eps^{-1}\log(\setsize\prob^{-1})))$ time. Each time we report an index $j$, we delete $O(\eps^{-4d}\log^{2d}(\setsize\prob^{-1}))$ points from $Q_j$. The overall query time is $O(\log^{4d+1}(\setsize\cdot\eps^{-1}\log(\setsize\prob^{-1}))+\out\cdot \eps^{-4d}\log^{2d}(\setsize\prob^{-1})\log^{4d+1}(\setsize\cdot\eps^{-1}\log(\setsize\prob^{-1})))$.
\end{proof}

\subsection{Approximate data structure for the \problemDI problem with a logical expression}
\label{subsec:percConj}
So far, we assumed that the (query) logical expression $\Pi$ contains only one predicate, $\pred_{\M_\rect,\interval}$. In this subsection, we extend our data structure to work for any constant number $\conj$ of predicates. Handling disjunction of predicates is straightforward using the designed data structure. For example, if the query consists of rectangles $\rect^{(1)},\ldots, \rect^{(\conj)}$ with intervals $\interval^{(1)},\ldots, \interval^{(\conj)}$ then by keeping track of the indexes we report (so that we do not report duplicates), we can handle every rectangle independently. We might encounter an index $i$ at most $O(\conj)=O(1)$ times so all our guarantees from the previous section hold. On the other hand, it is not clear how to handle conjunction of predicates. 

\paragraph{Main idea}
In this section, we describe a data structure that handles $\conj$ conjunctions of predicates. The main difference from the previous section is that we map all possible $\conj$ pairs of rectangles $(\rec,\hat{\rec})$ in $4\conj d$ dimensions instead of $4d$ dimensions that we had in the previous section. Given a query rectangle $\rect$ and $\conj$ predicates we define an orthant $\rect'$ in $4\conj d$ dimensions. The construction of the data structure is similar as in Subsection~\ref{subsec:approxPercRangePred}.

We only describe the data structure for general range-predicates. If we consider threshold-predicates the exponential complexities on $d$ are improved.

\paragraph{Data structure}
The first steps of constructing the data structure are identical with the data structure in the previous subsection.
Without loss of generality assume that all datasets and query rectangles lie in a bounded box $\mathcal{B}$.
For every $i\in[\setsize]$, we get an $(\eps+\delta)$-sample $\epsample_i\subset\Re^d$ by sampling $O(\eps^{-2}\log(\setsize\prob^{-1}))$ points from $\Synop_{\points_i}$ uniformly at random.
Let $\bar{\epsample}_i$ be the projection of $\epsample_i$ onto the $2\cdot d$ boundaries of $\mathcal{R}$.
For every $i\in [\setsize]$, we construct the set $\mathcal{R}_i$ that contains all combinatorially different hyper-rectangles defined by $\epsample_i\cup \bar{\epsample}_i$. Every hyper-rectangle $\rec\in \mathcal{R}_i$ is defined by their two opposite corners $\rec^-,\rec^+\in \Re^d$ such that $\rec^-_h\leq \rec^+_h$, for every $h\in[d]$, where $\rec^-_h, \rec^+_h$ are the $h$-coordinates of $\rec^-$ and $\rec^+$, respectively.
We construct all pairs of rectangles $\hat{\mathcal{R}}_i=\{(\rec, \hat{\rec})\mid \rec\in\mathcal{R}_i, \hat{\rec}\in \mathcal{R}_i,\rec\subseteq \hat{\rec}, \nexists \rec'\in \mathcal{R}_i: \rec\subset \rec'\subset\!\!\subset \hat{\rec}\}$.
Let $\tilde{\mathcal{R}}_i$ contain all subsets of $\conj$ pairs of rectangles in $\hat{\mathcal{R}}_i$. For every $\sigma=\{(\rec^{(1)},\hat{\rec}^{(1)}),\ldots (\rec^{(\conj)},\hat{\rec}^{(\conj)})\}\in \tilde{\mathcal{R}}_i$, let $$q_{\sigma}=(\rec_1^{(1)-},\ldots, \rec^{(1)-}_d,\hat{\rec}^{(1)-}_1,\ldots, \hat{\rec}^{(1)-}_d, \rec^{(1)+}_1,\ldots, \rec^{(1)+}_d, \hat{\rec}^{(1)+}_1,\ldots, \hat{\rec}^{(1)+}_d,\ldots, \rec_1^{(\conj)-},\ldots, \rec^{(\conj)-}_d,$$
$$\hat{\rec}^{(\conj)-}_1,\ldots, \hat{\rec}^{(\conj)-}_d, \rec^{(\conj)+}_1,\ldots, \rec^{(\conj)+}_d, \hat{\rec}^{(\conj)+}_1,\ldots, \hat{\rec}^{(\conj)+}_d)$$
be the point in $\Re^{4\conj d}$ defined by merging the opposite corners of $(\rec^{(r)}, \hat{\rec}^{(r)})\in \sigma$, for $r\in[\conj]$.
For every point $q_{\sigma}$, we assign $\conj$ weights, $w_{q_\sigma}=(\M_{\rec^{(1)}}(\epsample_i),\ldots, \M_{\rec^{(\conj)}}(\epsample_i))=
(\frac{|\rec^{(1)}\cap \epsample_i|}{|\epsample_i|},\ldots, \frac{|\rec^{(\conj)}\cap \epsample_i|}{|\epsample_i|})$.
Let $Q_i=\{q_{\sigma}\mid \sigma\in \tilde{\mathcal{R}}_i\}$ and $W_i=\{w_{q_{\sigma}}\mid q_{\sigma}\in Q_i\}$. Finally, we define $Q=\bigcup_{i\in [\setsize]}Q_i$ and $W=\bigcup_{i\in [\setsize]}W_i$. We construct a dynamic range tree $\rangetree$ over the (weighted) points $Q$ with their corresponding $\conj$-dimensional weights in $W$.

\paragraph{Query procedure}
We are given $\conj$ query rectangles $\{\rect^{(1)},\ldots, \rect^{(\conj)}\}$ in $\Re^d$ and $\conj$ intervals $\interval^{(1)}=[a_{\interval^{(1)}},b_{\interval^{(1)}}],\ldots \interval^{(\conj)}=[a_{\interval^{(\conj)}},b_{\interval^{(\conj)}}]$. Let $\rect^-, \rect^+$ be the two opposite corners of a rectangle $\rect$. We define the orthant $\rect'=[\rect^{(1)-}_1,\infty)\times \ldots\times [\rect^{(1)-}_d,\infty)
\times (-\infty,\rect^{(1)-}_1)\times\ldots\times (-\infty,\rect^{(1)-}_d)\times (-\infty,\rect^{(1)+}_1]\times\ldots\times(-\infty,\rect^{(1)+}_d]\times (\rect^{(1)+}_1,\infty)\times\ldots\times (\rect^{(1)+}_d,\infty)\times\ldots\times 
\rect^{(\conj)-}_1,\infty)\times \ldots\times [\rect^{(\conj)-}_d,\infty)
\times (-\infty,\rect^{(\conj)-}_1)\times\ldots\times (-\infty,\rect^{(\conj)-}_d)\times (-\infty,\rect^{(\conj)+}_1)\times\ldots\times(-\infty,\rect^{(\conj)+}_d]\times (\rect^{(\conj)+}_1,\infty)\times\ldots\times (\rect^{(\conj)+}_d,\infty)
$
in $\Re^{4\conj d}$ and the weight query rectangle $I'=[a_{\interval^{(1)}}-\eps-\delta,b_{\interval^{(1)}}+\eps+\delta]\times\ldots\times [a_{\interval^{(\conj)}}-\eps-\delta,b_{\interval^{(\conj)}}+\eps+\delta]$.
We repeat the following until no index is returned:
Using, $\rect'\times I'$ we query the range tree $\rangetree$. Let $u$ be the first (non-empty) canonical node it returns and let $q_{\sigma}$ be any arbitrary point stored in this node. Let $j$ be the index of the point set $\points_j$ that $q_{\sigma}$ was constructed for. We report $j$, and we delete from $\rangetree$ all points $\{q_{\sigma}\mid\rec^{(1)}\in \mathcal{R}_j\}$. We continue with the next iteration.
If there is no such canonical node or we do not find any $q_{\sigma}$ we stop the execution.
Let $J$ be the indexes we reported.
In the end of the process, we re-insert in $\rangetree$ all points we removed from $Q_j$ for every index $j\in J$.

\paragraph{Correctness}
The correctness of the query procedure follows from simple modifications of Lemmas~\ref{lem:techlem2},~\ref{lem:HH2},~ \ref{lem:tech1}, and~\ref{lem:tech2}.
Let $\Pi=\pred_{\M_{\rect^{(1)}},\interval^{(1)}}\land\ldots\land \pred_{\M_{\rect^{(\conj)}},\interval^{(\conj)}}$ be the logical expression consisting of the conjunction of $\conj$ range-predicates. It holds that $\query_\Pi(\sets)\subseteq J$, with probability at least $1-\prob$. Indeed, by extending the proof of Lemma~\ref{lem:tech1}, we have that: For an index $i\in \query_\Pi(\sets)$, let $\rec^{(i)}$ be the largest rectangle in $\mathcal{R}_i$ such that $\rec^{(i)}\subseteq \rect^{(i)}$. Let $q_{\sigma}$ be the point in $Q_i$ constructed by $\rec^{(1)}, \ldots, \rec^{(\conj)}$. From Lemma~\ref{lem:tech1}, it is straightforward to see that $q_{\sigma}\in \rect'\cap Q$ and $w_{q_{\sigma}}\in [a_{\interval^{(1)}}-\eps-\delta,b_{\interval^{(1)}}+\eps+\delta]\times\ldots\times [a_{\interval^{(\conj)}}-\eps-\delta,b_{\interval^{(\conj)}}+\eps+\delta]$ and the result follows. For every $j\in J$, let $q_\sigma\in Q_j$ be the point found in $\rect'$ at the moment we reported $j$. From Lemma~\ref{lem:techlem2}, we know that $\rec^{(\ell)}$ is the largest rectangle inside $\rect^{(\ell)}$ so the result follows by the definition of $\eps$-sample following the proof of Lemma~\ref{lem:tech2}.
Finally, by removing $\{q_{\sigma}\mid\rec^{(1)}\in \mathcal{R}_j\}$ we make sure that we will not report the index $j$ again, because any time that our query procedure finds $q_{\sigma}$, the rectangle $\rec^{(\ell)}$ must be the largest rectangle inside $\rect^{(\ell)}$, for every $\ell\in[\conj]$.

\paragraph{Analysis}
The next two lemmas follow straightforwardly by extending the proof of Lemmas~\ref{lem:space2}.
\begin{lemma}
\label{lem:spaceConj}
The constructed data structure has $$O(\setsize\cdot\eps^{-8\conj d}\log^{4\conj d}(\setsize\prob^{-1}) \log^{4\conj d-1+\conj}(\setsize\cdot\eps^{-1}\log(\setsize\prob^{-1})))$$ space and it can be constructed in $$O(\setsize\cdot\eps^{-8\conj d}\log^{4\conj d}(\setsize\prob^{-1}) \log^{4\conj d-1+\conj}(\setsize\cdot\eps^{-1}\log(\setsize\prob^{-1})))$$ time.
\end{lemma}

\begin{lemma}
The query procedure runs in $$O(\log^{4\conj d+\conj}(\setsize\cdot\eps^{-1}\log(\setsize\prob^{-1}))+\out\cdot \eps^{-4d}\log^{2d}(\setsize\prob^{-1})\log^{4\conj d+\conj}(\setsize\cdot\eps^{-1}\log(\setsize\prob^{-1})))$$ time.  
\end{lemma}

We set $\eps\leftarrow \eps/2$ and we get the final result.
\begin{theorem}
\label{thm:resConj}
Let $\{\Synop_{\points_1},\ldots,\Synop_{\points_\setsize}\}$ be the input to the $\problemDI$ problem, such that $\Synop_{\points_i}$ is a synopsis of a dataset $\points_i\subset \Re^d$, where $d\geq 1$ is a constant, with $\err_{\Synop_{\points_i}}(\F^d_{\square})\leq \delta$ for $\delta\in[0,1)$, for every $i\in[\setsize]$. Let $\eps,\prob\in(0,1)$ and $\conj> 1$ be three parameters.
A data structure of size $O(\setsize\cdot\eps^{-8\conj d}\log^{4\conj d}(\setsize\prob^{-1}) \log^{4\conj d-1+\conj}(\setsize\cdot\eps^{-1}\log(\setsize\prob^{-1})))$ can be constructed in $O(\setsize\cdot\eps^{-8\conj d}\log^{4\conj d}(\setsize\prob^{-1}) \cdot\log^{4\conj d-1+\conj}(\setsize\cdot\eps^{-1}\log(\setsize\prob^{-1})))$ time, such that
given $\conj$ query rectangles $\rect^{(1)},\ldots, \rect^{(\conj)}$ and $\conj$ query intervals $\interval^{(1)},\ldots, \interval^{(\conj)}$,
it returns a set of indexes $J$ such that $\query_\Pi(\sets)\subseteq J$, where $\Pi=\pred_{\M_{\rect^{(1)}},\interval^{(1)}}\land\ldots\land \pred_{\M_{\rect^{(\conj)}},\interval^{(\conj)}}$, and for every $j\in J$, $a_{\interval^{(\ell)}}-\eps-2\delta\leq \M_{\rect^{(\ell)}}(\points_j)\leq b_{\interval^{(\ell)}}+\eps+2\delta$, for $\ell\in[\conj]$, with probability at least $1-\prob$.
The query time is $O(\log^{4\conj d+\conj}(\setsize\cdot\eps^{-1}\log(\setsize\prob^{-1}))+\out\cdot \eps^{-4d}\log^{2d}(\setsize\prob^{-1})\log^{4\conj d+\conj}(\setsize\cdot\eps^{-1}\log(\setsize\prob^{-1})))$.

\end{theorem}
While in Theorem~\ref{thm:resConj} we only mention conjunction of predicates, the same result holds if we have any disjunction/conjunction of $\conj$ predicates.
If we assume that $\eps$ is an arbitrarily small constant, the success probability is high (at least $1-\frac{1}{\setsize}$, and $\conj$ is a constant, then the data structure has size $O(\setsize\cdot\log^{8\conj d-1+\conj}\setsize)=\O(\setsize)$, can be constructed in $O(\setsize\cdot\log^{8\conj d-1+\conj}\setsize)=\O(\setsize)$ time, and it has $O( \log^{4\conj d+\conj}\setsize+\out\cdot\log^{4\conj d+2d+\conj}\setsize)=\O(1+\out)$ query time.

\paragraph{Remark 1} The data structure in Theorem~\ref{thm:resConj} can be made dynamic under insertion or deletion of synopses. Using the update procedure of the dynamic range tree, if a synopsis $\Synop_\points$ is inserted/deleted
we can update our data structure in $O(\eps^{-8\conj  d}\log^{4\conj d}(\setsize\prob^{-1})\log^{4\conj d+\conj}(\setsize\eps^{-1}\log(\setsize\prob^{-1})))$ time.

\paragraph{Remark 2} If we do not know an upper bound $\delta$ on $\err_{\Synop_{\points_i}}(\F^d_{\square})$, then by slightly modifying the data structure above we can get the following result. Let $\err_{\Synop_{\points_i}}(\F^d_{\square})\leq \delta_i$ for unknown parameters $\delta_i\in[0,1)$.
Having the same complexities as in
Theorem~\ref{thm:resConj}, we report a set of indexes $J$ such that, if $a_{\interval^{(\ell)}}+\eps+\delta_i\leq \M_{\rect^{(\ell)}}(\points_i)\leq b_{\interval^{(\ell)}}-\eps-\delta_i$ for every $\ell\in[\conj]$, then $i\in J$, and if $j\in J$, then $a_{\interval^{(\ell)}}-\eps-\delta_j\leq \M_{\rect^{(\ell)}}(\points_j)\leq b_{\interval^{(\ell)}}+\eps+\delta_j$, for every $\ell\in[\conj]$. 

\section{Missing details from Section~\ref{sec:topK}}
\label{appndx:topk}
\begin{proof}[Proof of Lemma~\ref{lem:space4}]
    We have $|\net|=O(\eps^{-d+1})$ and each value $\gamma^{(i)}_{\vector}$ is computed in $\timeSynop_{\Synop_{\points_i}}$ time.
    Hence all points $\bigcup_{\vector\in \net}\Gamma_{\vector}$ are computed in $O(\eps^{-d+1}\cdot \sum_{i\in[\setsize]}\timeSynop_{\Synop_{\points_i}})$ time. Furthermore, each tree $\tree_{\vector}$ is computed in $O(\setsize\log \setsize)$ time and it uses $O(\setsize)$ space. Hence, the data structure has $O(\eps^{-d+1}\setsize)$ space and can be constructed in $O(\eps^{-d+1}\cdot \sum_{i\in[\setsize]}\timeSynop_{\Synop_{\points_i}}+\eps^{-d+1}\setsize\log\setsize)$ time.
\end{proof}

\paragraph{Remark 1} The data structure in Theorem~\ref{thm:resTopkAdd} can be made dynamic under insertion or deletion of synopses. If a synopsis $\Synop_\points$ is inserted/deleted
we can update our data structure in
$O(\timeSynop_{\Synop_{\points}}+\log\setsize)$.

\paragraph{Remark 2} If we do not know an upper bound $\delta$ on $\err_{\Synop_{\points_i}}(\vec{\F}^d_{k})$, then by slightly modifying the data structure from Theorem~\ref{thm:resTopkAdd} we can get the following result. Let $\err_{\Synop_{\points_i}}(\vec{\F}^d_{k})\leq \delta_i$ for unknown parameters $\delta_i\in[0,1)$.
Having the same complexities as in
Theorem~\ref{thm:resTopkAdd}, we report a set of indexes $J$ such that, if $\M_{\vec{u},k}(\points_i)\geq \intervalL+\eps+\delta_i$, then $i\in J$, and if $j\in J$, then $\M_\rect(\points_j)\geq \intervalL-\eps-\delta_j$. In other words, we might miss some datasets that satisfy the predicate $\pred_{\M_{\vec{u},k},\interval}$ close to the boundary $\intervalL$. More specifically, we might miss an index $i$ if $\M_{\vec{u},k}(\points_i)\in[\intervalL,\intervalL+\eps+\delta_i)$ and we always report an index $j\in J$ if $\M_{\vec{u},k}(\points_j)\geq \intervalL-\eps-\delta_j$.

\paragraph{Remark 3}
Our data structure in Theorem~\ref{thm:resTopkAdd} straightforwardly satisfies $O(\log\setsize)$ delay.

\subsection{Approximate data structure for the \problemDIk problem with a logical expression}
\label{subsec:topkpredm}
So far, we assumed that the (query) logical expression $\Pi$ contains only one predicate, $\pred_{\M_{\vec{u},k},\interval}$. In this subsection, we extend our data structure to work for any constant number $\conj$ of predicates. Handling disjunction of predicates is straightforward using the designed data structure, as we had for the \problemDI problem. On the other hand, it is not clear how to handle conjunction of predicates. In this section, we describe a data structure that handles $\conj$ conjunctions of predicates. 

\paragraph{Main idea}
The main difference from the previous subsection is that we construct a data structure for each subset of $\conj$ vectors in the $\eps$-net to compute an approximate solution for the \problemDIk problem.

We construct an $\eps$-net and for every vector $\vector$ in the $\eps$-net we use $\Synop_{\points_i}$ to get a value that approximates $\score_k(\points_i,\vector)$, for every $i\in[\setsize]$.
For every subset of $\conj$ vectors in $\net$, and every index $i\in[\setsize]$ we define a point in $\Re^\conj$ whose coordinates are the computed values on the subset of vectors.
We store these points to a range tree.
When we get query with $\conj$ query vectors, we find the closest vectors from the $\eps$-net and we use the corresponding range tree to identify the datasets whose scores lie in $\interval^{(1)},\ldots, \interval^{(\conj)}$.

\paragraph{Data structure}
We construct an $\eps$-net $\net$ on $\Re^d$.
For every $\vector \in\net$ and every $i\in[\setsize]$, we get a value $\gamma_{\vector}^{(i)}$ from $\Synop_{\points_i}$ as an approximation to $\score_k(\points_i,\vector)$.
Let $\Gamma_{\vector}=\{\gamma_{\vector}^{(i)}\mid i\in[\setsize]\}$.
For every subset $V=\{\vector_1,\ldots, \vector_\conj\}$ of $\conj$ vectors in $\net$, let $B_V=\{(\gamma_{\vector_1}^{(i)},\ldots,\gamma_{\vector_\conj}^{(i)})\mid i\in[\setsize]\}$ be a set of $\setsize$ points in $\conj$ dimensions.
For every subset $V$, we construct the (static) $\conj$-dimensional range tree $\rangetree_V$ for reporting queries over $B_V$.

\paragraph{Query procedure}
Let $\vec{u}^{(1)},\ldots, \vec{u}^{(\conj)}$ be the query vectors and $a_{\interval^{(1)}},\ldots, a_{\interval^{(\conj)}}$ be the thresholds. For every $\ell\in[\conj]$, we go through each vector in $\net$ and we compute $\vector_\ell=\argmin_{\vec{h}\in\net}||\vec{u}^{(\ell)}-\vec{h}||$. Let $V=\{\vector_1,\ldots, \vector_\conj\}$ and let $R=[a_{\interval^{(1)}}-\eps,1]\times\ldots \times [a_{\interval^{(\conj)}}-\eps,1]$ be a rectangle in $\Re^\conj$.
We run a reporting query on $\rangetree_{V}$ using the query rectangle $R$. Let $J$ be the set of all indexes reported.

\paragraph{Correctness}
Let $\Pi=\pred_{\M_{\vec{u}^{(1)},k},\interval^{(1)}}\land\ldots\land \pred_{\M_{\vec{u}^{(\conj)},k},\interval^{(\conj)}}$ be the logical expression consisting of the conjunction of $\conj$ range-predicates. Recall that $\query_\Pi(\sets)$ contains denotes the correct set of indexes we should report.

\begin{lemma}
    $\query_\Pi(\sets)\subseteq J$ and for every $j\in J$, $\score_k(\points_j,\vec{u}^{(\ell)})\geq a_{\interval^{(\ell)}}-2\eps-2\delta$, for every $\ell\in[\conj]$.
\end{lemma}
The correctness follows by straightforwardly extending the proof of Lemma~\ref{lem:tech3}.

\paragraph{Analysis}
\begin{lemma}
    The data structure has $O(\eps^{-\conj d+\conj}\setsize\log^{\conj-1}\setsize)$ space and can be constructed in\\ $O(\eps^{-d+1}\sum_{i\in[\setsize]}\timeSynop_{\Synop_{\points_i}} +\eps^{-\conj d+\conj}\setsize\log^{\conj-1}\setsize)$ time.
\end{lemma}
\begin{proof}
    We have $|\net|=O(\eps^{-d+1})$ and each value $\gamma^{(i)}_{\vector}$ is computed in $\timeSynop_{\Synop_{\points_i}}$ time.
    Hence all points $\bigcup_{\vector\in \net}\Gamma_{\vector}$ are computed in $O(\eps^{-d+1}\cdot \sum_{i\in[\setsize]}\timeSynop_{\Synop_{\points_i}})$ time. Furthermore, each tree $\tree_{V}$ is computed in $O(\setsize\log^{\conj-1} \setsize)$ time and it uses $O(\setsize\log^{\conj-1}\setsize)$ space. There are $O(\eps^{-\conj d +\conj})$ subsets of $\conj$ vectors in $\net$.

\end{proof}

\begin{lemma}
    For $\conj>1$, the query time is $O(\conj\eps^{-d+1}+\log^{\conj-1} (\setsize) +\out)$.
\end{lemma}
\begin{proof}
The query procedure spends $O(\conj\cdot\eps^{-d+1})$ time to compute $V=\{\vector_1,\ldots,\vector_\conj\}$. Then, it runs a reporting query on $\rangetree_{V}$ that takes $O(\log^{\conj-1} (\setsize)+\out)$ time. The result follows.
\end{proof}
We set $\eps\leftarrow \eps/2$ and we get the final result.
\begin{theorem}
\label{thm:resTopkAddConj}
Let $\{\Synop_{\points_1},\ldots,\Synop_{\points_\setsize}\}$ be the input to the $\problemDIk$ problem, such that $\Synop_{\points_i}$ is a synopsis of a dataset $\points_i\subset \Re^d$, where $d\geq 1$ is a constant, with $\err_{\Synop_{\points_i}}(\vec{\F}^d_{k})\leq \delta$ for $\delta\in[0,1)$ and $k\geq 1$, for every $i\in[\setsize]$. Let $\eps\in(0,1)$ and $\conj>1$ be two parameters.
A data structure of size $O(\eps^{-\conj d+\conj}\setsize\log^{\conj-1}\setsize)$ can be constructed in $O(\eps^{-d+1}\sum_{i\in[\setsize]}\timeSynop_{\Synop_{\points_i}} +\eps^{-\conj d+\conj}\setsize\log^{\conj-1}\setsize)$ time, such that, 
given $\conj$ query unit vectors $\vec{u}^{(1)},\ldots, \vec{u}^{(\conj)}\in\Re^d$ and $\conj$ thresholds $a_{\interval^{(1)}},\ldots, a_{\interval^{(\conj)}}\geq 0$,
it returns a set of indexes $J$ such that $\query_\Pi(\sets)\subseteq J$, where $\Pi=\pred_{\M_{\vec{u}^{(1)},k},\interval^{(1)}}\land\ldots\land \pred_{\M_{\vec{u}^{(\conj)},k},\interval^{(\conj)}}$,
and for every $j\in J$,
$\M_{\vec{u}^{(\ell)},k}(\points_j)\geq a_{\interval^{(\ell)}}-\eps-\delta$ for every $\ell\in[\conj]$.
The query time is $O(\conj\eps^{-d+1}+\log^{\conj-1} (\setsize) +\out)$.
\end{theorem}
While in Theorem~\ref{thm:resTopkAddConj} we only mention conjunction of predicates, the same result holds if we have any disjunction/conjunction of $\conj$ predicates.
If we assume that $\eps$ is an arbitrarily small constant and $\conj$ is a constant then the data structure has size $O(\setsize\log^{\conj-1}\setsize)=\O(N)$, can be constructed in 
$O(\sum_{i\in[\setsize]}\timeSynop_{\Synop_{\points_i}} +\setsize\log^{\conj-1}\setsize)=\O(\sum_{i\in[\setsize]}\timeSynop_{\Synop_{\points_i}}+\setsize)$ time, and it has $O(\log^{\conj-1}(\setsize)+\out)=\O(1+\out)$ query time.

\paragraph{Remark 1} The data structure can be made dynamic under insertion or deletion of synopses. If a synopsis $\Synop_\points$ is inserted/deleted
we can update our data structure in $O(\eps^{-d+1}\timeSynop_{\Synop_{\points}}+\eps^{-\conj d+\conj}\log^{\conj}\setsize)$ time. 

\paragraph{Remark 2} If we do not know an upper bound $\delta$ on $\err_{\Synop_{\points_i}}(\vec{\F}^d_{k})$, then by slightly modifying the data structure above we can get the following result. Let $\err_{\Synop_{\points_i}}(\vec{\F}^d_{k})\leq \delta_i$ for unknown parameters $\delta_i\in[0,1)$.
Having the same complexities as in
Theorem~\ref{thm:resConj}, we report a set of indexes $J$ such that, if $\M_{\vec{u}^{(\ell)},k}(\points_i)\geq a_{\interval^{(\ell)}}+\eps+\delta_i$ for every $\ell\in[\conj]$, then $i\in J$, and if $j\in J$, then $\M_{\vec{u}^{(\ell)},k}(\points_j)\geq a_{\interval^{(\ell)}}-\eps-\delta_j$, for every $\ell\in[\conj]$.

\end{document}